\documentclass[smallextended,envcountsect,envcountsame,envcountreset,12pt]{svjour3}       

\smartqed  
\usepackage{graphicx}
\usepackage{amsmath}
\usepackage{amsfonts}
\usepackage{amssymb}
\usepackage{graphicx}
\begin{document}

\title{Bounding the Size of a Network Defined By Visibility Property}
\subtitle{}


\author{Andreas D. M. Gunawan         \and
        Louxin Zhang 
}


\institute{Andreas Dwi Maryanto Gunawan \and Louxin Zhang \at
	Department of Mathematics, National University of Singapore, Block S17, Lower Kent Ridge Road 10,
Singapore 119076. \\
	\email{a0054645@u.nus.edu}       
}

\date{Received: date / Accepted: date}

\maketitle

\begin{abstract}
Phylogenetic networks are mathematical structures for modeling and 
visualization of reticulation processes in the study of evolution. 
Galled networks,  reticulation visible networks, nearly-stable networks and stable-child networks are the four classes of phylogenetic networks that are recently introduced to study the topological and algorithmic aspects of phylogenetic networks. We prove the following results.   
\begin{itemize}
\item[ ~~~~(1)] A binary galled network with $n$ leaves has at most $2(n-1)$ reticulation nodes.
\item[ ~~~~(2)] A binary nearly-stable network with $n$ leaves has at most $3(n-1)$ reticulation nodes.
\item[ ~~~~(3)] A binary stable-child network with $n$ leaves has at most $7(n-1)$ reticulation nodes.
\end{itemize}

\end{abstract}

\section{ Introduction}
\label{intro}

Reticulation processes refer to the transfer of genetic material between living organisms in a non-reproduction manner. 
Horizontal gene transfer is  believed to be a highly significant reticulation process occurring between single-cell organisms (Doolittle and Bapteste 2007; Treangen and Rocha 2011).
Other reticulation processes include introgression, recombination and hybridization (Fontaine et al. 2015; McBreen and Lockhart 2006;  Marcussen et al. 2014). In the past two decades, phylogenetic networks have often been seen for the modeling and visualization of  reticulation processes (Gusfield 2014; Huson et al. 2011).

Galled trees, galled networks, reticulation visible networks  are  three of the popular classes of phylogenetic networks introduced to study the combinatorial and algorithmic perspectives of phylogenetics (Wang et al. 2001; Gusfield et al. 2004;  Huson and Kloepper 2007; Huson et al. 2011).  
Reticulation visible networks include galled trees and galled networks. They are tree-based
(Gambette et al.  2015).
The tree-based networks are introduced by Francis and Steel (2015) recently.

It is well known that the number of internal nodes in a phylogenetic tree with $n$ leaves is $n-1$. In contrast, an arbitrary phylogenetic network with 2 leaves can have as many internal nodes as possible. 
Therefore, one interesting research problem is how large a phylogenetic network in a particular class can be.  For example, it is well known  that a tree-child network with $n$ leaves has $3(n-1)$ non-leaf nodes at most. A regular network with $n$ leaves has 
$2^{n}$ nodes at most (Willson 2010).  To investigate whether or not the tree containment problem is polynomial time solvable, surprisingly,  Gambette et al. (2015) proved that a reticulation visible network with 
$n$ leaves has at most $9(n-1)$ non-leaf nodes. The class of nearly-stable networks was also introduced in their paper.   They also proved the existence of a linear upper bound on the number of reticulation nodes in a nearly-stable network. 

In the present paper, we establish the tight upper bound for the size of a network defined by a visibility property using a sub-tree technique that was introduced in Gambette et al. (2015). The rest of this paper is divided into six sections.  Section~\ref{sec:Basic} introduces  concepts and notation that are necessary for our study. 
Recently,  Bordewich and Semple (2015) proved that there are at most $3(n-1)$ reticulation nodes in a reticulation visible network.  In Section~\ref{sec:Stable}, we present a different proof of the $3(n-1)$ tight bound for reticulation visible networks.   
Section~\ref{sec:Galled}  proves  that  there are at most $2(n-1)$ reticulation nodes in a galled networks with $n$ leaves. 
Section \ref{sec:NearlyStable} and \ref{sec:StableChild} establish the tight upper bounds for the sizes of nearly-stable and stable-child networks, respectively. In Section~\ref{conc}, we conclude the work with a few remarks.

\section{ Basic concept}
\label{sec:Basic}
\subsection{Phylogenetic Networks}

An acyclic digraph is a simple connected digraph with no directed cycles.

Let $D=({\cal V}(D), {\cal E}(D))$ be an acyclic digraph and  let $u$ and $v$ be two nodes in $D$. If $(u, v)\in {\cal E}(D)$,  it is called an \emph{outgoing} edge of $u$ and \emph{incoming edge}  of $v$; $u$ and $v$ are said to be the tail and head of the edge. The numbers of incoming and outgoing edges of a node are called its \emph{indegree} and \emph{outdegree}, respectively. 
$D$ is said to be \emph{rooted} if there is a unique node $\rho(D)$  with indegree 0; $\rho(D)$ is called the \emph{root} of $D$.  Note that  in a rooted acyclic digraph there exists a directed path from the root to every other node. 

For  $E \subseteq \mathcal{E}(D)$,  $D-E$ denotes  the digraph with the same node set and the edge set $ \mathcal{E}(D) - E$. 
For $V \subseteq \mathcal{V}(D)$,  $D-V$ denotes  the digraph with the node set  $\mathcal{V}(D)-V$ and the edge set 
$\{(u,v) \in \mathcal{E}(D) | u\not\in V \mbox{ and } v \notin V  \}$.
If $D'$ and $D''$ are  subdigraphs of $D$,  $D' + D''$ denotes the subdigraph  with the node set $\mathcal{V}(D') \cup \mathcal{V}(D'')$ and the edge set $\mathcal{E}(D') \cup \mathcal{E}(D'')$.

A \textit{phylogenetic network} on a finite set  of taxa, $X$,  is a rooted acyclic digraph in which  each non-root node has either indegree 1 or outdegree 1 and there are exactly $|X|$ nodes of outdegree 0 and indegree 1, called \emph{leaves}, that correspond one-to-one with the taxa in the network.  

In a phylogenetic network, a node is called a \emph{tree node} if it is either the root or a node having indegree one; it is called a \emph{reticulation node} if its indegree is greater than one.
Note that leaves are tree nodes and  a tree node may have both indegree and outdegree one. A non-leaf node is said to be {\it internal}. 
A phylogenetic network without reticulation nodes is simply a  \textit{phylogenetic tree}. 

For a phylogenetic network $N$, we use the following notation:
\begin{itemize}
\item   $\rho(N)$:  the  root of $N$.
\item  $\mathcal{V}(N)$:  the set of nodes.
\item  $\mathcal{T}(N)$: the set of tree nodes.
\item  $\mathcal{R}(N)$: the set of reticulation nodes.
\item  $\mathcal{E}(N)$: the set of edges.
\item  $\mathcal{L}(N)$:  the set of leaves.
\end{itemize}
For two nodes $u, v$ in $\mathcal{V}(N)$,  if $(u, v)\in \mathcal{E}(N)$, $u$ is said to be a \emph{parent} of $v$ and, equivalently,  $v$ is a \emph{child} of $u$. 
In general, if there is a directed path from $u$ to $v$, $u$ is an  \emph{ancestor} of $v$ and $v$ is a \emph{descendant} of $u$. We sometimes say that $v$ is below $u$ when $u$ is an ancestor of $v$.

Let $P$ and $Q$ be two simple  paths from $u$ to $v$ in $N$. We use  $\mathcal{V}(P)$ and $\mathcal{V}(Q)$ to denote their node sets, respectively.   They are \emph{internally disjoint} if $\mathcal{V}(P) \cap V(Q) = \{u, v\}$.

Finally, a   phylogenetic network is \textit{binary}, if its root has outdegree 2 and indegree 0,  all internal nodes have degree 3, and all the leaves have indegree one. Here, we are interested in how large a binary phylogenetic network can be. 

In the rest of the paper, a binary phylogenetic work is simply called a network and a phylogenetic tree a tree. For sake of convenience for discussion, we also add an open edge entering the root of a network.

\subsection{Visibility Properties}

 A node $v$ in a  network is \emph{visible} (or stable) with respect to a leaf $\ell$ if $v$ is in every path from the network root to $\ell$. We say $v$ visible if it is visible with respect to some leaf in the network.

\begin{lemma}
	\label{Prop1} 
	Let $N$ be a network and $N'$  a subnetwork of $N$ with the same root and leaves as $N$. 
	Then, a node is visible in $N'$ if it is visible in $N$. Equivalently, a node is not visible in $N$ if it is not visible in $N'$.
\end{lemma}
\begin{proof}
	Suppose $v \in \mathcal{V}(N)$ is visible with respect to a leaf $\ell$ in $N$. 
For each path $P$ from $\rho(N')$ to $\ell$, since it is also a path from $\rho(N)$ to $\ell$ in $N$, it must pass through $v$. Thus, $v$ is also visible with respect to the same leaf in $N'$.
\qed
\end{proof}

\emph{Reticulation visible} networks are  networks in which reticulation nodes are all visible (Huson et al. 2011). They are also called stable networks by Gambette { et al.} (2015).

A network is \emph{galled} if  every reticulation node $r$ has an ancestor $a$ such that 
there are two disjoint tree paths from $a$ to $r$ (Huson and Kloepper 2007).  Here, a path is tree path if its internal nodes are all tree nodes in the network.  Galled networks are reticulation visible and  are also known as level-1 networks. 

 \emph{Nearly-stable} networks are  networks in which for every pair of nodes  $u$ and $v$, either $u$ or $v$ is visible if 
 $(u, v)$ is an edge (Gambette et al. 2015).
 
 \emph{Stable-child} networks  are  networks in which every node has a visible child.
 
 Tree-based networks comprise another interesting class of networks that is introduced recently (Francis and Steel 2015). A network is \emph{tree-based} if it can be obtained from a tree with the same leaves by the insertion of a set of edges between different edges in the tree.  
 
\begin{theorem}
\label{stable_1}
{\rm (Gambette et al. 2015)}
For every reticulation-visible network $N$, there exists a subset of edges $E\subseteq {\cal E}(N)$ such that $E$ contains exactly an incoming edge for each reticulation node and $N-E$ is a subtree with the same leaves as $N$. 
\end{theorem}

Theorem~\ref{stable_1} indicates that every reticulation visible network is tree-based. 
However, nearly-stable networks and stable-child networks  are not necessarily  tree-based.

We finish this section by presenting a technical lemma that will frequently be  used in establishing the tight upper bound on the size of a network in each of the four classes defined above.

\begin{lemma}
	\label{lemma22}
	Let $N$ be a network,  $u \in {\cal V}(N)$, and $R$ be a finite set of  reticulation nodes  below $u$.
If  each $r\in R$ has a parent $p(r)$ such  that either {\rm (a)}  $p(r)$ is  below another $r'$ in $R$, or {\rm (b)} there is a path from $\rho (N)$ to $p(r)$ that avoids $u$, 
		then there exists a path from $\rho (N)$ to $\ell$ avoiding $u$ for every leaf $\ell$ below a reticulation node $r\in R$.
\end{lemma}
	\begin{proof}
		Let  $\ell \in {\cal L}(N)$. Assume $\ell$ is below some $r_1\in R$. Then,  there is a path $P(r_1, \ell)$ from $r_1$ to $\ell$ that avoids $u$.
		Since $R$ is finite and $N$ is acyclic, there exists a series  of reticulation nodes, $r_1, r_2, \cdots, r_k$ such that:
		
(i) 	each $r_j$ has a parent $p_j$ below $r_{j+1}$ for $j=1, 2, \cdots, k-1$,  and

(ii) the node $r_k$ has a parent $p_k$ such that there is a path $P(\rho(N), p_k)$ from $\rho(N)$ to $p_k$ that avoids $u$. 
	
		Since $p_j$ is below $r_{j+1}$, there exists a path 
		$P(r_{j+1}, p_{j})$ from $r_{j+1}$ to $p_{j}$ for each $j< k$.  Since $N$ is acyclic and $r_{j+1}$ is below $u$, the path $P(r_{j+1}, p_{j})$ avoids $u$. 
		Concatenating these paths, we obtain the following path
		$$P(\rho (N), p_{k})+(p_k, r_k) + P(r_k, p_{k-1})+(p_{k-1}, r_{k-1}) + \cdots +
		(p_1, r_1)+P(r_1, \ell)$$
		from $\rho$ to $\ell$ that avoids $u$. 
		\qed
	\end{proof}

\section{ Reticulation visible networks}
\label{sec:Stable}

Gambette et al. (2015) proved that there are at most $4(n-1)$ reticulation nodes in a reticulation visible network with $n$ labeled leaves. 
On the other hand, there are as many as  $3(n-1)$ reticulations in the reticulation visible network in Figure~\ref{Fig1}.
So, what is the tight upper bound on the number of reticulation nodes? Interestingly,  $3(n-1)$ is the tight upper bound, which was independently proved by Bordewich and Semple (2015) using the induction approach.  Here, we present an alternative proof to illustrate our approach. 

\begin{figure}[!b] 
	\begin{center}
\includegraphics*[width=0.8\textwidth]{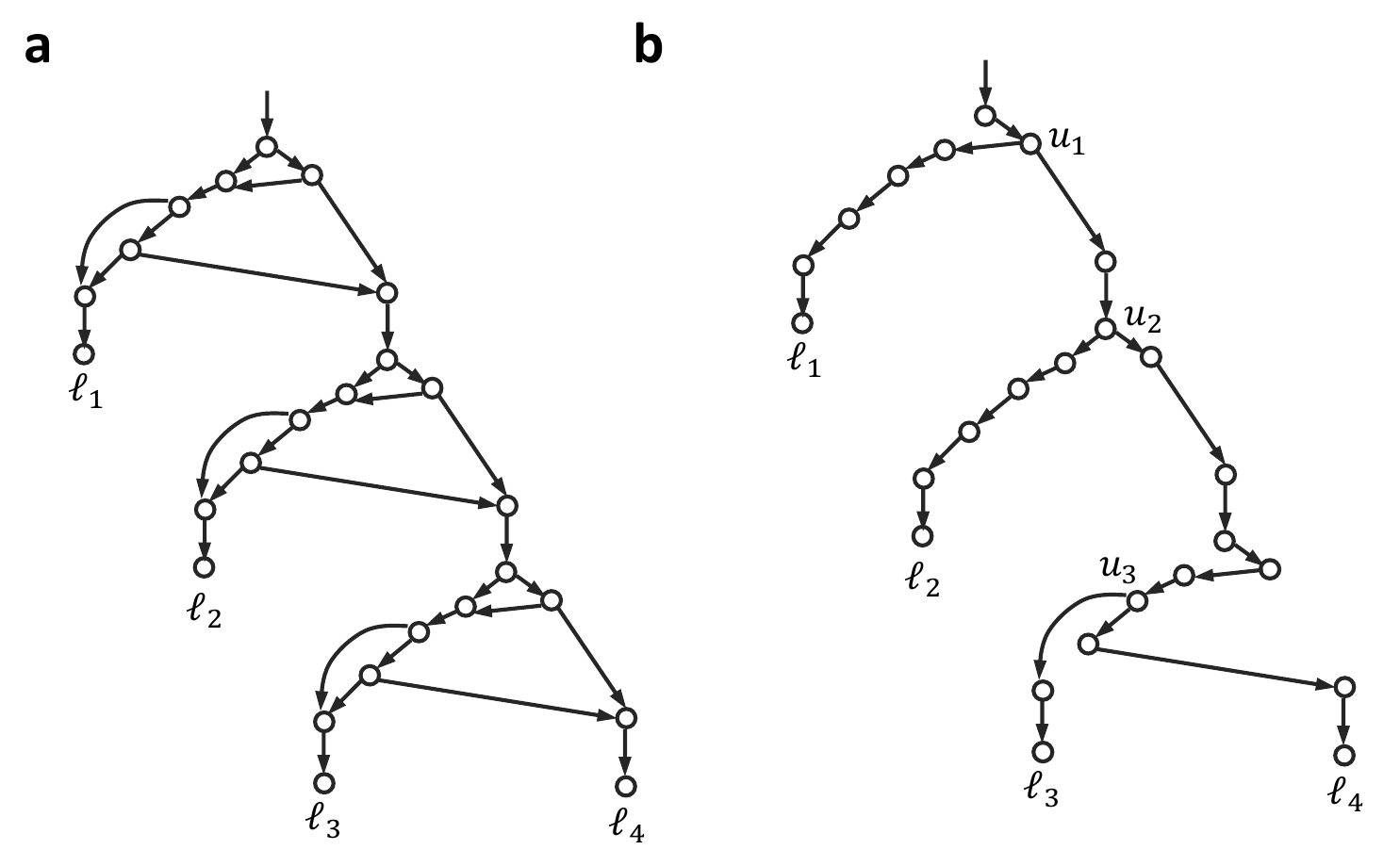}	
	\end{center}
	\caption{ ({\bf a}) A reticulation visible network with 4 leaves that has as many reticulation nodes as possible. ({\bf b}) A subtree of the network that has the same leaves, in which only $u_1, u_2, u_3$ are of degree 3 \label{Fig1}}
\end{figure}


Given a reticulation visible network $N$ with $n$ leaves, we 
let $E$ be a set of edges such that $N-E$ is a subtree with the same root and
leaves as $N$ (Theorem~\ref{stable_1}). 
Since $N-E$ has $n$ leaves, there are exactly $n-1$  nodes of degree 3. Thus,  there are $2n-2$  paths whose internal nodes are of degree 2, starting at a degree-3 node and terminating  at either another node of degree 3 or a leaf.  Let these $2n-2$ paths be  $P_1, P_2, \cdots, P_{2n-2}$.

The edges of $N-E$ not in $\cup_{1\leq i\leq 2n-2} \mathcal{E}(P_i)$ make up a path $P_0$ that contains the root $\rho (N)$ (Figure~\ref{Fig1}). If $\rho (N)$ is of degree 2, $P_0$ passes through $\rho(N)$ and terminates at a degree-3 node.
If $\rho(N)$ is of degree 3,  $P_0$  is simply the open edge entering $\rho (N)$.
Altogether, these $2n-1$ paths are called the {\it trivial paths} of $N-E$. Note that  $\mathcal{E}(N-E) = \uplus_{0\leq i\leq 2n-2} \mathcal{E}(P_i)$.

It is not hard to see that, for each edge in $E$,  its head and tail  are both found in these trivial paths.  An edge $(u,v) \in E$ is called a \emph{cross} edge if $u\in {\cal V}(P_i)$ and $v\in {\cal V}(P_j)$ for $i\neq j$; it is called   \emph{non-cross} edge otherwise.  The facts in the following proposition appear in the proof of Theorem 1 in  Gambette et al (2015).

\begin{proposition}
\label{31}
\begin{itemize}
 \item[{\rm (1)}] No two cross edges $e\in E$ have their heads in the same trivial path in $N-E$.  
 \item[{\rm (2)}] For each non-cross edge
 $(u, v)$ such that $u, v\in P_i$ for some $i\geq 0$, 
 there are at least one cross edge $(w,x)\in E$ such that $w$ is between $u$ and $v$ in $P_i$. 
\end{itemize} 
\end{proposition}

For a cross edge $(u, v)$ such that $u\in {\cal V}(P_i)$ and $v\in {\cal V}(P_j)$($i \neq j$), we say $(u, v)$ \emph{leaves} $P_i$ and \emph{enters} $P_{j}$. For a non-cross path $(u, v)$ and a cross edge $(w, x)$, if $u$ and $v$ are in $P_{i}$ and $w$ is a node between $u$ and $v$ in $P_i$, we say $(u, v)$ \emph{jumps over}  $(w, x)$. 

It is trivial to see that no cross edge enters the trivial path $P_0$. 
 Proposition~\ref{31} suggests that $E$ contains at most $2n-2$ cross edges and thus at most $2n-2$ non-cross edges. By Theorem~\ref{stable_1}, $|{\cal R}(N)|=|E|\leq 4n-4$.

To obtain the tight upper bound $3n-3$ for $|{\cal R}(N)|$, we define the cost $c(e)$ of a cross edge $e\in E$ as:
  \begin{equation*}
  \label{cost_def}
   c(e)=\begin{cases}2\;\;\mbox{if there is an non-cross edge jumping over the tail of $e$}, \\1\;\;\mbox{otherwise.}
        \end{cases}
  \end{equation*}
 We will charge the cost of a cross edge  to the trivial path  it enters and 
  call it the \emph{weight} of the trivial path.
 If no cross edge enters a trivial path, the  weight of this trivial path is set to be 0.  By Proposition~\ref{31}, the weight of a trivial path is at most 2.
 We use $w(P_i)$ to denote the weight of a trivial path $P_i$, $0\leq i\leq 2n-2$.

\begin{figure}[!b] 
	\begin{center}
		\includegraphics*[width=0.95\textwidth]{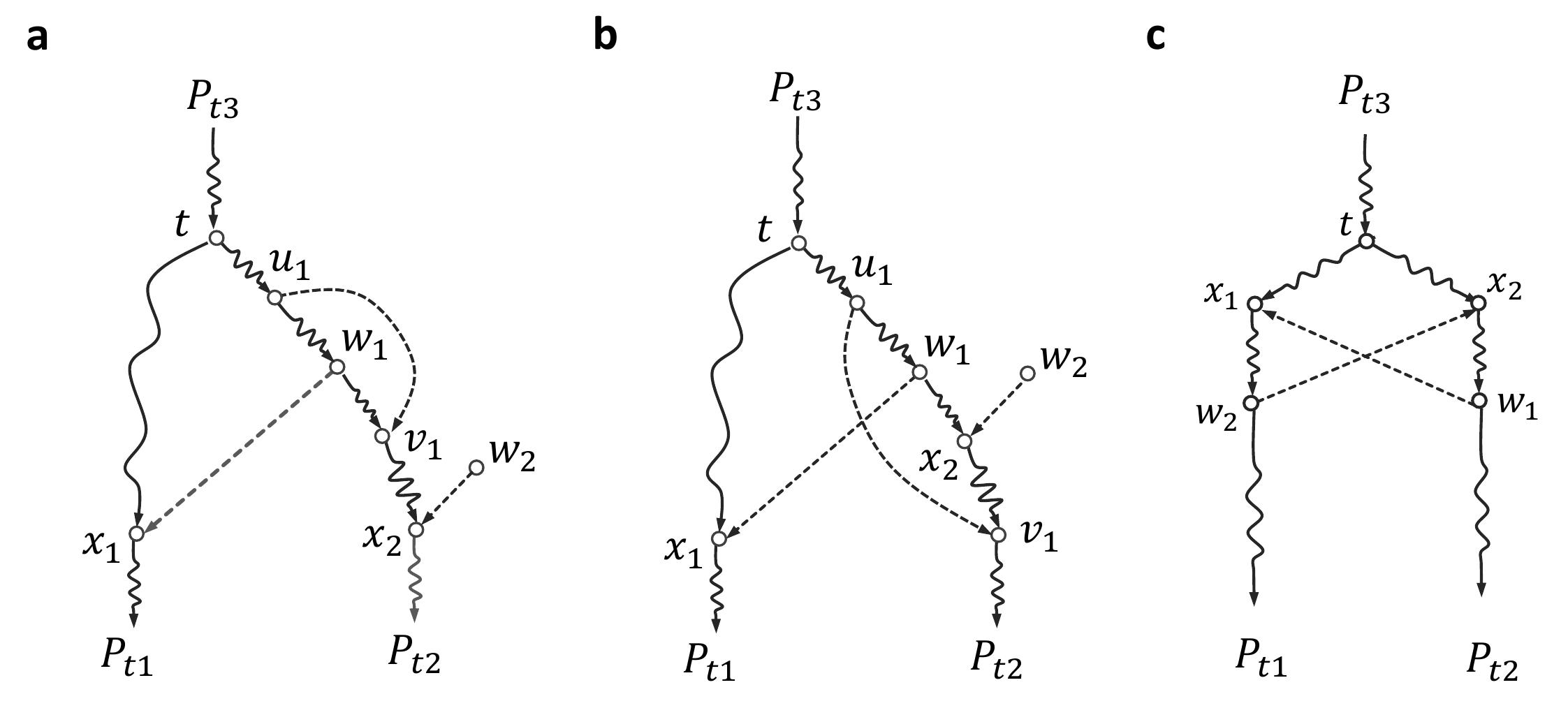}	
	\end{center}
	\caption{Three cases that are considered in the proof of Proposition~\ref{Stable_3}. $P_{t1}, P_{t2}, P_{t3}$ are the three trivial paths incident to a degree-3 node $t$; $(w_i, x_i)$ is the cross edge ending at $v_i$ in $P_{ti}$ and the non-cross edge $(u_i, v_i)$ jumps over $(w_i, x_i)$ for $i=1, 2$.  Here, $(u_2, v_2)$ is not drawn.  ({\bf a}) $w_1$ and $v_1$ are both between $t$ and $x_2$ in $P_{t2}$.
	({\bf b}) $w_1$ is between $t$ and $x_2$ in $P_{t2}$, but $v_1$ is below $x_2$ in $P_{t2}$. 
	({\bf c})  The node $w_1$ is below $x_2$ and  $w_2$ is below $x_1$. 
	This case is impossible to occur, as there is a directed cycle. 
	\label{Fig2}}
\end{figure}


For an internal  node $t$ of degree 3 in $N-E$, we use  $P_{t3}$ to denote the trivial path entering $t$ and
$P_{t1}, P_{t2}$ to denote the two trivial paths leaving $t$.

\begin{proposition} \label{Stable_3}
	Let $t$ be a degree-3 node in $N-E$.
	
{\rm (i)} 	If  $P_{t3}\neq P_0$ and  $w(P_{t1})=w(P_{t2})=2$, then $w(P_{t3})=0$.  

{\rm (ii)} If $P_{t3}=P_0$, then $P_{t3}=0$ and $w(P_{t1})+w(P_{t2})\leq 3$.
\end{proposition}
\begin{proof}
For sake of simplicity, we let $T=N-E$ and use $P_{T}(z', z'')$ to denote the unique path from $z'$ to $z''$ for a node $z'$ and a descendant $z''$ of  $z'$ in $T$. $N$ and $T$ have the same root and leaves. The common root of $N$ and $T$ is written $\rho$. 

(i) Assume that $P_{t3}\neq P_0$ and  $w(P_{t1})=w(P_{t2})=2$. Then, there exists a cross edge $(w_j,x_j)$ entering $P_{tj}$ and a  non-cross edge $(u_j,v_j)$ jumping over $(w_j,x_j)$ for each  $j =1,2$. 
 We  shall prove that
$w(P_{t3})=0$ by showing that no cross edge enters $P_{t3}$.

Assume $w_1$ is between $t$ and $x_2$ in $P_{t2}$. 
When $v_1$ is below $w_1$ in $P_{T}(t, x_2)$ (Figure~\ref{Fig2}a), there are two cases.
If 
$w_2$ is in the path $P_{t1}$ or below it, then, $P_{T}(\rho, w_2)$ does not pass $v_1$. 
If $w_2$ is not below $t$, $P_{T}(\rho, w_2)$ does not pass $t$ and so $v_1$. 
Therefore, by Lemma~\ref{lemma22}, there is path from $\rho$ to every leaf below $x_2$ that does not pass $v_1$. For any leaf $\ell$ not below $x_2$ in $T$, 
$P_{T}(\rho, \ell)$ avoids $v_1$.  Hence, $v_1$ is a reticulation node in $N$, but not visible. This is a contradiction.  

When $v_1$ is below $x_2$ in $T$ (Figure~\ref{Fig2}b),  $v_1$ is below $x_2$ as a reticulation node. Since $P_{T}(\rho, u_1)$ does not pass $x_2$, by Lemma~\ref{lemma22}, there is a path from $\rho$ to a leaf below $x_2$ that does not pass $x_2$. For any leaf $\ell$ not below $x_2$ in $T$, $P_{T}(\rho, \ell)$ avoids $x_2$. Hence, $x_2$ is not visible, a contradiction. 

We have proved that $w_1$ is not between $t$ and $x_2$ in the tree path $P_{t2}$.
By symmetry, $w_2$ is not between  $t$ and $x_1$ in  $P_{t1}$. 

Assume there is a cross edge enters $P_{t3}$.  Let $r$ be the lowest reticulation node in $P_{t3}$. Then, $w_1$ and $w_2$ are both not in $P_{T}(r, t)$. Otherwise, either $v_1$ or $v_2$ is between $r$ and $t$, contradicting that $r$ is the lowest reticulation node in $P_{t3}$. 
Combining this fact with that $w_j$ is not between $t$ and $x_{3-j}$ in $T$ for $j=1, 2$, we conclude that either $w_j$ is below $x_{3-j}$ or there is a path from $\rho$ to $w_j$ not passing $r$ for each $j=1, 2$. Hence, by 
Lemma~\ref{lemma22}, $r$ is not visible with respect to any leaf below $r$.
For any leaf $\ell$  not below $r$ in $T$, the tree path $P_{T}(\rho, \ell)$
avoids $r$. Hence, $r$ is not visible, a contradiction. 

 We have proved that $w(P_{t3})=0$. 

(ii) If $P_{t3}=P_0$, then $t$ is an ancestor of any other degree-3 node in $N-E$. 
 Since $N$ is acyclic, there does not exist $(u, v)\in E$ such that $u\in P_i$ for some $i>0$ and $v\in P_0$. Hence, $w(P_0)=0$. 

 Assume on the contrary 
 the weights of $P_{t1}$ and $P_{t2}$ are both 2. Then, $w_j$ is not between $t$ and $x_{3-j}$ for $j=1, 2$, proved above.  If $w_1$ or $w_2$ is in $P_0$, the lowest reticulation in $P_0$ is not visible, a contradiction. Otherwise,  $w_1$ is below $x_2$ and $w_2$ is below $x_1$, implying a cycle in $N$ (Figure~\ref{Fig2}c). This is a contradiction. Hence, either $P_{t1}$ or $P_{t2}$ has weight less than 2.
$\Box$
\end{proof}

\begin{theorem} Let $N$ be a reticulation visible network with $n$ leaves. Then,
	$$|\mathcal{R}(N)| \leq 3(n-1). $$
\end{theorem}
\begin{proof}
Let $V$ denote the set of $(n-1)$ internal nodes of degree 3 in $N-E$. Note that 
any trivial path other than $P_0$ starts with a node in $V$.   Define:
$$V_{k}=\{ v\in V\;|\; w(P_{v1})+w(P_{v2})=k\}$$
for $0\leq k\leq 4$.
Clearly, $V_k$'s are pairwise disjoint and hence 
$$|V_0|+|V_1|+|V_3| +|V_3|+|V_4|=|V|=n-1.$$

When $v\in V_4$, 
$w(P_{v1})=w(P_{v2})=2$. By Proposition~\ref{Stable_3}, $P_{v3}\neq P_0$.
Let  $p(v)$ be the start node of $P_{v3}$ for each $v\in V_4$.
Again, by Proposition~\ref{Stable_3}, $p(v)\in V_{0}
\cup V_{1}\cup V_{2}$. It is 
clear that under the map $p(\cdot)$, at most two nodes in $V_4$ are mapped to the same node in $V_0$, and different nodes in $V_4$ are mapped to different nodes in $V_1\cup V_2$. Thus, 
$|V_4|\leq 2|V_0|+ |V_1|+|V_2|$. Since $w(P_0)=0$, the inequality implies that
	$$\begin{array}{rll}
	|\mathcal{R}(N)|
	&= &\sum _{v\in V}[w(P_{v1})+w(P_{v2})]\\
&	=& |V_1|+2|V_2|+3|V_3|+ (3|V_4|+|V_4|) \\
&	\leq& 2(|V_0|+|V_1|) + 3(|V_2|+|V_3|+|V_4|)\\
&\leq & 3(|V_0|+|V_1|+|V_2|+|V_3|+|V_4|)\\
&=& 3(n-1),
\end{array}$$
where the first inequality is derived from the substitution of $2|V_0|+ |V_1|+|V_2|$ for $|V_4|$.
	\qed
\end{proof}

	\section{ Galled networks}
	\label{sec:Galled}
	
Galled networks form a subclass of reticulation visible networks (Huson et al. 2011). 	In this section,  we shall show that there are at most $2(n-1)$ reticulations  in a galled network with $n$ leaves.  Given that the galled network shown in Figure~\ref{Fig3}a has exactly  $2(n-1)$ reticulations,  $2(n-1)$ is the tight bound on  the number of reticulation nodes in a galled network with $n$ leaves.
	
	
	\begin{figure}[!b] 
		\begin{center}
			\includegraphics*[width=0.8\textwidth]{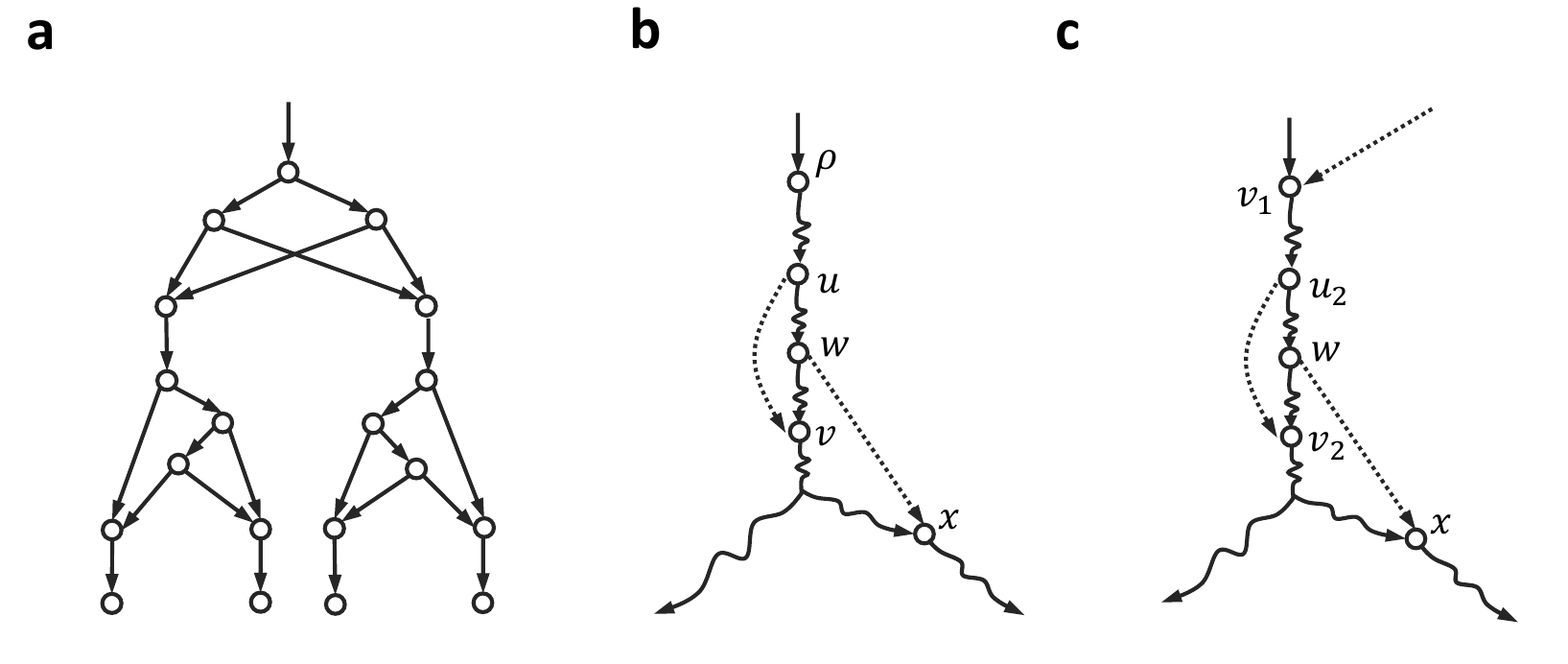}	
		\end{center}
		\caption{ ({\bf a}) A galled network with 4 leaves that has as many reticulation nodes as possible. ({\bf b}) and ({\bf c}) are two cases considered in the proof of  Theorem~\ref{GN_1}: 
there is a non-cross edge $(u, v)$ in $E$ such that $u$ and $v$ are in $P_0$, and  there is a cross-edge edge $(u_1, v_1)$ and a non-cross edge
$(u_2, v_2)$ both ending at a node 
 in a  trivial path other than $P_0$, where $u_1$ is not  drawn. In ({\bf b}) and ({\bf c}),    
solid straight and curve arrows   represent edges and paths in $N-E$, respectively;  round dot arrows represent edges in $E$, respectively \label{Fig3}}
	\end{figure}

	\begin{theorem} \label{GN_1}
		For a galled network $N$ with $n$  leaves,
		$|\mathcal{R}(N)| \leq 2(n-1).$
	\end{theorem}
	\begin{proof}
	Let $N$ be a galled network with $n$ leaves and let $\rho=\rho(N)$.
Since $N$ is  reticulation visible, by Theorem~\ref{stable_1}, there is a set of edges $E$ such that (a) $E$ contains exactly one incoming edge for each reticulation node and (b) $N-E$ is a subtree with the same leaves as $N$. 

 We use the same notation as in Section 3. $P_0$ denotes the  trivial path whose first edge is the open edge entering $\rho$; 
$P_1, ...., P_{2n-2}$ denote the other $2n-2$ trivial paths in $N-E$. We prove the result by showing that $E$ does not contain any non-cross edges and 
only one cross-edge can end at a node in each $P_i$ for $i>0$. 

  If $P_0$ contains only  the open edge entering $\rho$, there is no edge in $E$ that enters $P_0$. We first prove that this fact is also true even if $P_0$ contains other edges below $\rho$.   
	
  Since $N$ is acyclic and there is a directed path from the end of $P_0$ to a node in $P_i$ for any $i>0$, there is no cross edge $(u,v) \in E$ such that  $u$ is in $P_i$ and $v$ is in $P_0$. 
 
 If there is a non-cross edge $(u,v)$ such that $u,v$ are in $P_0$ (Figure~\ref{Fig3}b), we let $w$ be the other child of $u$ in $P_0$.  Then, $w$ must be a tree node such that $(w,x)\in E$, where   $x$ is a reticulation node in some trivial path $P_i$, $i>0$. (If $w$ is a reticulation, it is not visible, a contradiction.)
Since $N$ is galled and  $x$ is a reticulation node, there exist
two  paths $P'$
and $P''$ from a common tree node to $x$ in $N$ such that (i) they are internally disjoint and (2) 
$x$ is the only reticulation  node in them. 
Note that no edges in $E$ other than $(w, x)$ can appear in $P'$ and $P''$. Otherwise, either $P'$ and $P''$ contains another reticulation node. 
Thus, $P'+ P'' - (w, x)$ is a subtree of $N-E$. This implies that 
 one of $P'$ and $P''$ is the single edge $(w, x)$ and the other is 
$P_{N-E}(w, x)$, the unique path from $w$ to $x$  in the tree $N-E$. This is impossible, as the reticulation node $v$ is in $P_{N-E}(w, x)$.

We have shown that there is no edge in $E$ that enters $P_0$.
Next, we show that there is at most one edge in $E$ that enters $P_i$ for each $1\leq i\leq 2n-2$.

Assume  that $(u_1,v_1)$ and $ (u_2,v_2)$ are two edges in $E$ such that $v_2$ is below $v_1$ in some $P_i$ ($i>0$) (Figure~\ref{Fig3}c). Then, $(u_2,v_2)$ must be a non-cross edge and $u_2$ is also below $v_1$. (Otherwise,   $v_1$ is not visible.)  Again, by Fact (2) in  Proposition~\ref{31}, there is a cross edge $(w,x)$ such that $w$ is between $u_2$ and $v_2$ in $P_i$ and $x$ is in $P_j$, $j\neq i$. 
	Since $x$ is a reticulation node and $N$ is galled, 
there are two internally disjoint paths $P'$ and $P''$ from a common tree node to $x$ in which  any nodes other than $x$ are a tree node. 
If $P' + P''$ contains an  edge in $E$ other than $(w, x)$,  the head of the edge is a reticulation node  and appears in either $P'$ or $P''$, a  contradiction. Hence,  $P'+ P''-(w, x)$ is a subtree of $T$.  Without loss of generality, we may assume $P'$ contains $(w, x)$. That is, $(w, x)$ is the last edge of $P'$. 
		Note that $v_1, u_2, w, v_2$ are all nodes in $P_i$, ordered from top to bottom. 
If $P'$ contains more than one edge in $T$, 
it must pass through $v_1$, a contradiction. If $P'$ is equal to $(w, x)$, then $P''$ must pass through $v_2$, a contradiction. Therefore, there is at most one edge in $E$ whose head is in each trivial path $P_i$, $i>0$. 

In summary, there are $2n-2$ trivial paths other than $P_0$ and there is at most one edge  in $E$ entering each of them. Hence, 
		$|\mathcal{R}(N)| = |E| \leq 2(n-1).$
		\qed
	\end{proof}

\section{ Nearly-stable network}
	\label{sec:NearlyStable}

	In this section we will give a tight bound for the number of reticulations in a nearly-stable network. The class of nearly-stable networks is different from the class of reticulation visible networks, but surprisingly the tight upper bound is also $3(n-1)$. The network shown in Figure~\ref{Fig4}a is an example for a nearly-stable network with $3(n-1)$ reticulations. We need the following fact, proved by Gambette et al. (2015).

	\begin{proposition} \label{NS_1}
		Let $N$ be a nearly-stable network with $n$ leaves. There exists a set $E$ of edges such that {\rm (a)} $N-E$ is a reticulation visible  subnetwork over the same leaves as $N$, and {\rm (b)} $E$ contains exactly one incoming edge for each  reticulation node that is not visible in $N$.
	\end{proposition}

	\begin{figure}[!b] 
		\begin{center}
			\includegraphics*[width=0.9\textwidth]{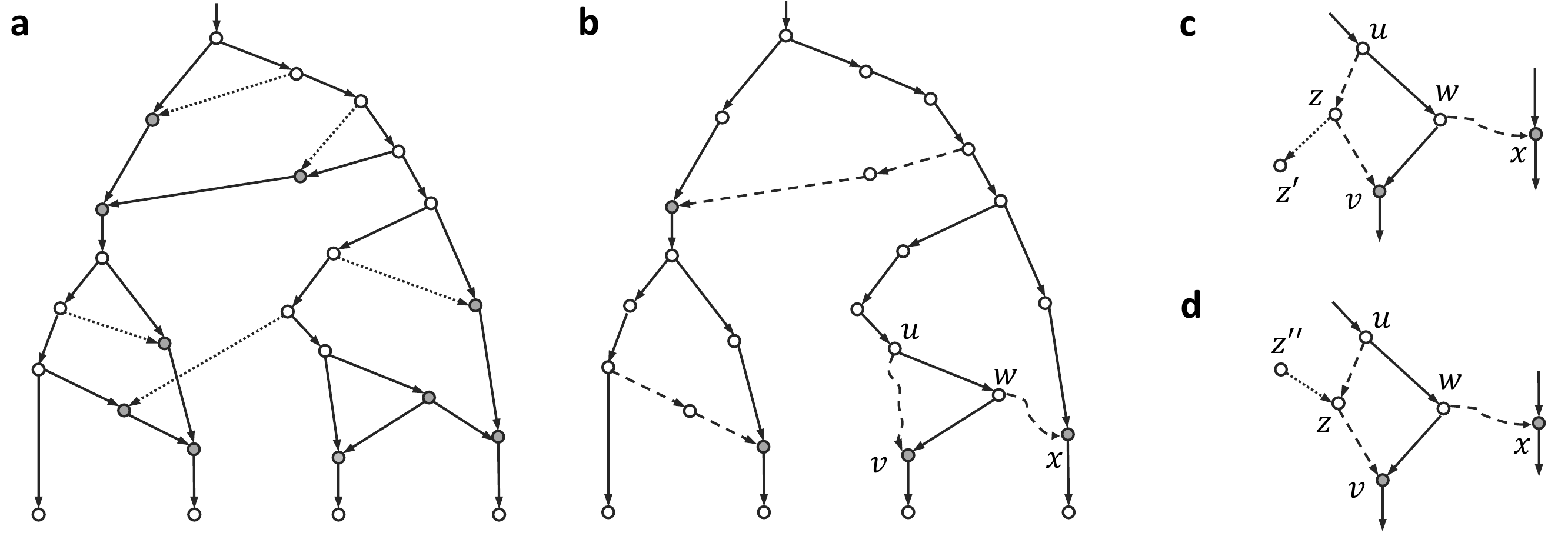}	
		\end{center}
		\caption{ ({\bf a}) A nearly-stable network $N$ with four leaves. It has nine reticulations (shaded circles), five of which are not visible.
The round dot edges are those removed to obtain the reticulation visible network $N'$ in ({\bf b}).  The dashed paths in $N'$ are the cross and non-cross paths removed to obtain a subtree with the same leaves as $N$. 
 (\textbf{c}) and (\textbf{d}) are two cases considered in the proof of the part (c) in Lemma~\ref{NS_2}: 
  a non-cross path from $u$ and $v$ contains a tree node $z$ of $N$, and it contains a reticulation node $z$ of $N$  \label{Fig4}}
	\end{figure}	

	Let $E$ be the set of edges satisfying the two properties in Proposition~\ref{NS_1} and let 
	$N' = N - E$  (Figure~\ref{Fig4}a).  The edges in $E$ are said to be \emph{{\rm NS}-edges}.  We remark that  $N'$ is a subdivision of a binary reticulation visible network. That is, the  reticulation visible network can be obtained from $N'$ by replacing some paths whose internal nodes are of degree 2 with  directed edges with the same orientation.   Hence, $N'$ contains degree-2 nodes if $E$ is not empty. 
	
	For a path $P$, we use $\mathcal{IV}(P)$ to denote the set of its internal nodes. Since $N'$ is a subdivision of a binary reticulation visible network with the same leaves as $N$, by Theorem~\ref{stable_1}, there is a set $\mathcal{P}$ of paths  in ${N'}$ such that
	(i) each path $P \in \mathcal{P}$ is  from a degree-3 tree node to a visible reticulation  node in $N'$ and its internal nodes are all of degree-2 in $N'$, and (ii) ${N'} - \cup_{P \in \mathcal{P}} \mathcal{IV}(P) - \cup_{P \in \mathcal{P}} {\cal E}(P)$ is a subtree with the same  leaves as $N$. 

Let $T={N'} - \cup_{P \in \mathcal{P}} \mathcal{IV}(P) - \cup_{P \in \mathcal{P}} {\cal E}(P)$. $T$ is obtained from the removal of the internal nodes and edges of the paths in ${\cal P}$. 
	We can classify the paths in $\mathcal{P}$ as \emph{cross paths} and \emph{non-cross paths} accordingly as in Section 3 (Figure~\ref{Fig4}b).

	\begin{lemma}\label{NS_2} 
Let $N$ be a nearly-stable network and let $E$, $N'$, $T$ and  $\mathcal{P}$ be defined above.  

{\rm (a)}  Every internal node in a path in $\mathcal{P}$ is not visible in $N$. 

{\rm (b)} Each cross path in $\mathcal{P}$ consists of either a single edge or two edges in $N$.

{\rm (c)} Each non-cross path  in $\mathcal{P}$ is simply an edge in $N$.

{\rm (d)} If $P$  is a cross path in $\mathcal{P}$
from $w$ to $x$ and $P'$ is a non-cross path in $\mathcal{P}$ from $u$ to $v$  such that  $w$ is between $u$ and $v$ {\rm (}Figure~\ref{Fig4}b{\rm )}, then $P$ and $P'$ are both a single edge in $N$. 

{\rm (e)} Every two distinct paths in $\mathcal{P}$ are node disjoint. 
	\end{lemma}
	\begin{proof} We remark that  $P_{T}(x, y)$ denotes the
unique path from $x$ to $y$ for any two nodes $x$ and $y$ in $T$. 
	
	(a)  Let $P$ be a path in $\mathcal{P}$ and let  $y$ be an internal node of it.  For any leaf 
	$\ell\in {\cal L}(T)$, the unique  path $P_{T}(\rho, \ell)$ 
does not pass $y$ in $T$. Hence,  $y$ is not visible in $N$. 

(b) If there are two or more internal nodes in a path in $\mathcal{P}$, 
by (a), they are  consecutive  and  not visible  in $N$, contradicting that $N$ is nearly-stable.

(c) We use $\rho$ to denote the root of $N$, which is also the root of $N'$ and $T$. Let $P$ be a non-cross path between $u$ and $v$, where $u$ and $v$ are in some path $P_i$ in $T$. 
Note that $P_{T}(u, v)$ is a sub-path of $P_i$  and is
 internally disjoint from $P$. 
 By Fact (2) in Proposition~\ref{31}, there is an internal node $w$ in $P_{T}(u, v)$ that is the start node of a cross path 
 $P(w, x)$ in $\mathcal{P}$.
 
First,  any node $y$ between $u$ and $v$ in $P_{T}(u, v)$ is not visible. This is because for any network leaf $\ell$ not below $v$ in $T$, $P_T(\rho, \ell)$ does not pass through $y$, and  for any network leaf $\ell$ below $v$ in $T$, $P_T(\rho, \ell)-P_T(u, v) + P$ is a path not passing through $y$.  Therefore,  $w$ must be the unique internal node of $P_{T}(u, v)$.
That is, $w$ is the child of $u$ and the parent of $v$ in $P_{T}(u, v)$.

Assume that $P$ is not an edge in $N$. By (a), there is a unique degree-2 node $z$ between $u$ and $v$ in $P$.   We consider the following two cases.

If $z\in {\cal T}(N)$ (Figure~\ref{Fig4}c),  then  the other outgoing  edge  $(z, z')$  had been removed 
to obtain $N'$. That is, 
$(z, z')\in E$. By the definition of $E$,
 $z'$ is a reticulation node and not visible in $N$. That $z$ and $z'$ are both not visible  contradicts that $N$ is 
nearly-stable.

If $z\in {\cal R}(N)$ (Figure~\ref{Fig4}d), then the other incoming edge $(z'', z)$  had been removed to obtain $N'$. 
Note that $z''\neq u$ and $z''\neq w$, as $w$ has degree 3 in $N'$. 
In addition, $z''$ is not an internal node of a path in $\mathcal{P}$. (Otherwise, by (a),  $z$ and $z''$ are both not visible). So $z''$ is a node in $T$.  Clearly $z''$ is not below $v$ and hence not below $w$ in $T$. (Otherwise $N$ has a cycle.)  Hence, $P_T(\rho,z'')$ does not pass through $u$. 

Consider a network leaf 
$\ell \in {\cal L}(N)$. If it is not below $v$, then $P_T(\rho, \ell)$ does not pass through $u$. If $\ell$ is below $v$, then  $P_{T}(\rho, z'') +(z'', z) + (z,v) +P_{T}(v, \ell)$ is a path not passing through $u$ in $N$. Therefore,  $u$ is not visible. That $u$ and $w$ are both not visible in $N$ contradicts that $N$ is nearly-stable. 

(d) By the proof of (c),  $P'$ is a single edge in $N$ and $w$ is the only node in $P$ and not visible. Thus $P$ must be an edge in $N$. (Otherwise by (a) $w$ and its child in $P$ are  not visible, contradicting that $N$ is nearly-stable.)

(e) It can be easily derived from the definition of the cross path. 
 \qed
	\end{proof}

	Let $C\in \mathcal{P}$ be a cross path from $w$ to $x$. Then, $x$ is a visible  reticulation node in $N$. It may have as many as two reticulation parents that are not visible.  Let $U_x = parent(x) \cap \mathcal{UR}(N)$, where $parent(x)$ is the set of all parents of $x$ and $\mathcal{UR}(N)$ is the set of all  reticulation nodes that are not visible in $N$. $|U_x|=0, 1,$ or $2$.  Define the cost of $C$ as:
	\begin{equation}
	\label{weight_def}
	c(C)=\begin{cases}2+|U_x|\;\;\mbox{if there is a non-cross edge jumping over $w$}, \\1+|U_x|\;\;\mbox{otherwise,}
	\end{cases} 
	\end{equation}
	where $2$ is used to count $x$ and the other child of $w$ which is a visible  reticulation node if there is a non-cross edge jumping over $w$. 
	
	As in Section~\ref{sec:Stable}, we let $P_0$ denote the trivial path whose first edge is the incoming edge to $\rho$
	and let $P_1,...,P_{2n-2}$ denote the other $2n-2$ trivial paths in $T$.
	We charge the cost of a cross path to the trivial path $P_i$ in $T$ in which the cross path enters and call it the weight of $P_i$.  The weight of $P_i$ is denoted by $w(P_i)$.
	If a trivial path does not contain any end node of the cross paths in $\mathcal{P}$,  its weight is set to be 0.
	
	Each visible reticulation node contributes to at least one unit of weight. By the definition of nearly-stable networks, any reticulation node that is not visible must have a visible reticulation node  as its child, and by the proof of Lemma~\ref{NS_2} (c), any reticulation node that is not visible in $N$ must be in some $U_x$, $x$ being the end node of a cross path, so it also contributes to at least one unit weight. Therefore, $|\mathcal{R}(N)| \leq \sum_{i = 0}^{2n-2} w(P_i)$.  To bound this, we first establish a useful lemma.

	As in Section~\ref{sec:Stable},   we  use $P_{t3}$ to denote the trivial path entering $t$ and $P_{t1}$, $P_{t2}$ to denote  the trivial paths leaving $t$ for a node $t$ of degree 3 in $T$.
	
	\begin{figure}[!t] 
		\begin{center}
			\includegraphics*[width=0.5\textwidth]{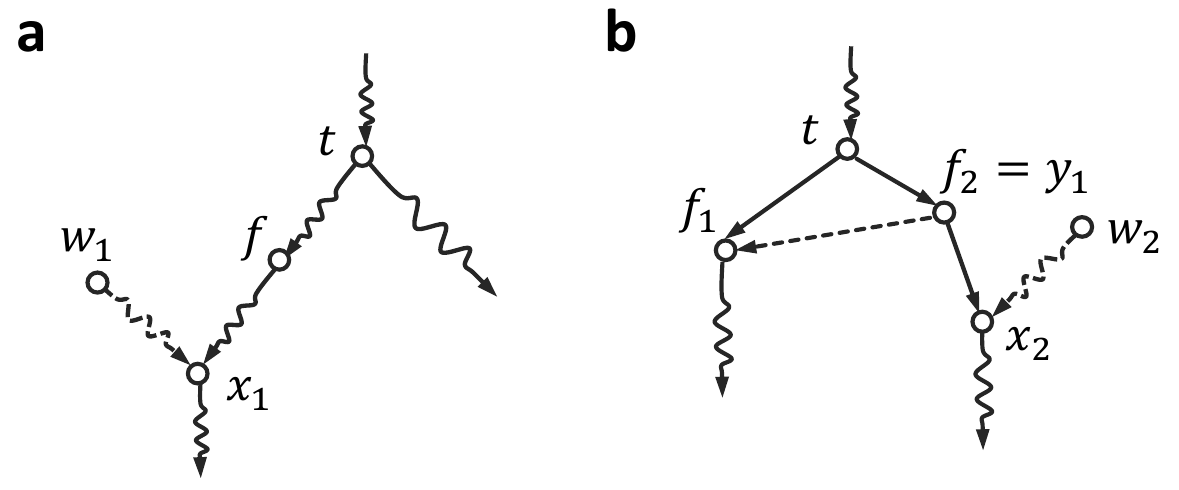}	
		\end{center}
		\caption{The two cases considered in  the proof of  Lemma~\ref{NS_33}. Solid arrows and curves represent the edges and paths in $T$, while square dot arrows and curves represent the removed edges and paths. $f_1$ is the reticulation child of $t$ in a trivial path $P_{t1}$ leaving $t$. $f_2$ is the child of $t$ in the trivial path $P_{t2}$.
			The path from $w_2$ to $x_2$ is a cross path entering $x_2$.
			({\bf a}) A cross path from $w_1$ to $x_1$ and $f$ is a node between $t$ and $x_1$, where  $x_1$ is in  $P_{t1}$.
			({\bf b})	The unique tree node $f_2$ between $t$ and $x_2$ is also a parent of $f_1$ in $N$ \label{Fig5}}
	\end{figure}
	
	\begin{lemma}
		\label{NS_33}
		Let $P_{tj}$ be a trivial path defined above and let $C_j$ be a cross path from $w_j$ to $x_j$, where  $x_j$ is in $P_{tj}$ and  $j \in\{ 1, 2\}$. Define $j'=3-j$.
		
		{\rm (a)}   The tree path $P_{T}(t, x_j)$ consists of  either a single edge or two edges in $P_{tj}$.  If $P_{T}(t, x_j)$ includes two edges, the internal node is not visible.
		
		{\rm (b)} If there exists a  node $f_j$ between $t$ and $x_j$ in $P_{tj}$ and $f_j$ is a reticulation node in $N$, then no cross-path enters the other trivial path $P_{tj'}$. 
		
		{\rm (c)} Assume that $C_j$ contains an internal node that is a reticulation node in $N$. If there is a cross-path $C_{j'}$ from $w_{j'}$ to $x_{j'}$ such that  $x_{j'}$ is in $P_{tj'}$,   then $w_j$ is not in
		$P_{T}(t, x_{j'})$. 
	\end{lemma} 
	\begin{proof} Note that  $\{j, 3-j\}=\{1, 2\}$ for $j=1, 2$. Without loss of generality, we may assume that  $P_{tj}=P_{t1}$ and $P_{tj'}=P_{t2}$, that is $j=1$ and $j'=3-j=2$. 
		
		(a) Let $f$ be a node between $t$ and $x_1$ in $P_{t1}$
		(Figure~\ref{Fig5}a) and let $\ell$ be a leaf in $N$. If  $\ell$ is not below $x_1$ in $T$, 
		the path $P_T(\rho,\ell)$ does not pass through $f$.
		
		Let $\ell$ be a leaf below $x_1$ in $T$. Since $w_1$ is not in $P_{t1}$ in $T$, the tree path $P_{T}(\rho, w_1)$ does not pass $f$. By  Lemma~\ref{lemma22}  there is a path from $\rho$ to $\ell$  that avoids $f$. 
		Therefore, $f$  is not visible.  
		
		Since $N$ is nearly-stable, there is at most one  node 
		in $P_{T}(t, x_1)$, as each  internal node is not visible.

		(b) Suppose on the contrary, there is a cross path $C_2$ from $w_2$ to $x_2$ entering 
		$P_{t2}$, where $x_2$ is in $P_{t2}$.  By  (a), $x_2$ is a child of $t$ or there is a unique node $f_2$ between $t$ and $x_2$ in $P_{t2}$.  We first show that $t$ is not visible in $N$.
		
		If $x_2$ is a child of $t$ or there is a node $f_2$ in $P_{T}(t, x_2)$ such that $f_2$ is a reticulation node in $N$,  $t$ has two reticulation children in $N$. By Lemma~\ref{lemma22}, $t$ is not visible. 
		
		If $P_{t2}$ contains a node $f_2$ between $t$ and $x_2$  in $N$,   $(f_2, f_1)$ must not be an edge in $N$. Otherwise, as shown in Figure~\ref{Fig5}b,   $f_1$ and $f_2$ are then  not visible, contradicting that $N$ is nearly-stable. 
		
		Let $(y_1, f_1)$ be the edge removed from $f_1$ in the process of transforming $N$ to $N'$.  Since  $y_1 \neq f_2$, either $y_1$ is below $x_2$ or there is a path from $\rho$ to $y_1$ that avoids $t$. 
		
		Since $w_2$ is in another trivial path and there is no node between $t$ and $f_1$ in $P_{t1}$,  $w$ is either 
		below $f_1$  or the path $P_{T}(\rho, w_2)$ does not pass $t$.
		
		Since the reticulation nodes $f_1, x_2$ are below $t$ and satisfy the condition in Lemma~\ref{lemma22}, there is a path from $\rho$ to $\ell$ that avoids $t$ for any leaf $\ell$ below $f_1$ or $x_2$.
		For any leaf  $\ell$ below neither $f_1$ nor $x_2$, it is not below $t$ and the path $P_{T}(\rho, \ell)$ does not pass through  $t$.  Therefore, $t$ is also not visible.
		
		The fact that $t$ and $f_1$ are both not visible contradicts that $N$ is nearly-stable. This implies that there is no cross path entering $P_{t2}$. 
		
		(c) If $x_2$ is the child of $t$ in $P_{t2}$, the case is trivial. 
		
		Assume that there is an internal node $f_2$ between $t$ and $x_2$ in $P_{t2}$. By the fact (a), $f_2$ is not visible.  If $w_1=f_2$, then $w_1$ and its child in $P_1$ are both not visible, contradicting $N$ is nearly-stable network.
		\qed
	\end{proof}

\begin{figure}[!bh] 
	\begin{center}
		\includegraphics*[width=0.9\textwidth]{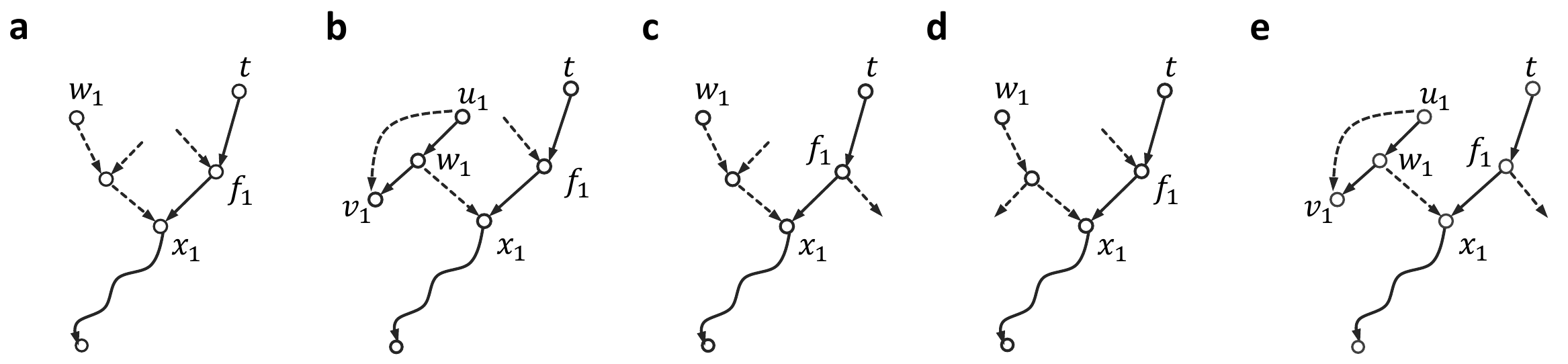}	
	\end{center}
	\caption{  
		({\bf a})-({\bf b})	Two types of trivial paths of weight 3.
		({\bf c})-({\bf e}) Three types of trivial paths of weight 2. Solid arrows and curves represent the edges and paths in $T$, while   square dot arrows and curves represent the removed edges and paths. The  path from $w_1$ to $x_1$ is the cross path ending at a node in a trivial path leaving $t$
		\label{Fig6}}
\end{figure}	

\begin{proposition} \label{NS_4}
	For an internal node $t$ of degree $3$ in $T$,
	
	{\rm (a)} $w(P_{t1})\leq 3$ and $w(P_{t2})\leq 3$. 
	
	{\rm (b)} if  $w(P_{tj})=3$, then  $w(P_{t(3-j)})=0$, where $j\in \{1, 2\}$.
	
	{\rm (c)} if $P_{t3} \neq P_0$ and  $w(P_{t1})=w(P_{t2})=2$, then $w(P_{t3})=0$. Moreover,  assume $p(t)$ is the degree-3 ancestor of $t$ such that  $P_{t3} = P_{p(t)1}$. Then $w(P_{p(t)2})\leq 2$.
	
	{\rm (d)} if $P_{t3} = P_0$, then $w(P_{t3}) = 0$ and $w(P_{t1})+w(P_{t2})\leq 3$.
\end{proposition}	
\begin{proof}
	(a) We only prove that  $w(P_{t1})\leq 3$. 
	If there is no non-cross edge jumping over the start node of the cross path entering $P_{t1}$, by Eqn.~(\ref{weight_def}), the weight of $P_{t1}$ is at most 3.

	If there is a non-cross edge jumping over the start node $w_1$ of the cross-path $C_1$ ending at a node $x_1$ in $P_{t1}$, by the fact (d) of  Lemma~\ref{NS_2},  $C_1$ is equal to 
	the single edge $(w_1, x_1)$.  Therefore, $x_1$ has at most one reticulation parent, which is in $P_{t1}$ if exists. By Eqn.~(\ref{weight_def}), 
	$w(P_{t1})\leq 3$.
	
	(b) Assume  $w(P_{t1})=3$. Then, there is a cross path $C_1$ from $w_1$ to $x_1$ where $x_1$ is in $P_{t1}$. If there is no non-cross edge jumping over $w_1$,  by Eqn.~(\ref{weight_def}), $x_1$ has  two reticulation parents (Figure~\ref{Fig6}a). 
	
	If there is a non-cross edge jumping over $w_1$,  by the fact (d) of Lemma~\ref{NS_2}, 
	$C_1$ is equal to a single edge $(w_1, x_1)$, and by Eqn.~(\ref{weight_def}), $x_1$ has
	one reticulation parent $f_1$ in $P_{t1}$ (Figure~\ref{Fig6}b). 
	By the fact (b) of Lemma~\ref{NS_33}, there is no cross path that enters $P_{t2}$, implying  	$w(P_{t2})=0$.

	(c) Assume $P_{t3} \neq P_0$ and  $w(P_{t1})=w(P_{t2})=2$. Let $C_j$ be the cross path from $w_j$ to $x_j$, with $x_j$ in $P_{tj}, j =1,2$.  Since $w(P_{t1})=w(P_{t2})=2$,  by the fact (b) of Lemma~\ref{NS_33},  there is no reticulation node between $t$ and $x_j$ for each $j$. Hence, 
	for each $j$, either  the parent of $x_j$ in  $C_j$ is a reticulation node and not visible (Figure~\ref{Fig6}c), or there is a non-cross edge $(u_i,v_i)$ jumping over $w_i$ (Figure~\ref{Fig6}e).

	
	By the facts (a) and (b)  of Lemma~\ref{NS_33}, either $x_j$ is the child of $t$ in $P_{tj}$ or there is a tree node $f_j$ between $t$ and $x_j$ in $P_{tj}$ for $j\in \{1, 2\}$. 
	
	Assume that there is a tree node $f_j$ between $t$ and $x_j$ in $P_{tj}$, $j\in \{1, 2\}$. Let $j'=3-j$.  
	If $C_{j'}$ has an internal node that is a reticulation, by the fact (c) of Lemma~\ref{NS_33}, $w_{j'} \neq f_j$. 
	
	If there is a non-cross edge jumping over $w_{j'}$,  by the fact (d) of Lemma~\ref{NS_2},  that $w_{j'}=f_j$ implies that 
	the endpoints of the non-cross edge are also between $t$ and $x_j$. This is impossible, as there is only $f_j$ between $t$ and $x_j$.  Therefore, $w_{j'} \neq f_j$. Similarly,  $w_{j} \neq f_{j'}$.
	
	We have proved that for $j=1,2$, $w_j$ is not between $t$ and $x_{j'}$. Thus,  $w_j$ is either below $x_{j'}$ or there is a path from $\rho$ to $w_j$ that does not pass $t$. Therefore, by Lemma~\ref{lemma22}, there is a path from $\rho$ to $\ell$ not passing through $t$ for any leaf below either $x_1$ or $x_2$. For any leaf $\ell$ below neither $x_1$ nor $x_2$, since it is not below $t$ in $T$,  $P_{T}(\rho, \ell)$ does not contain $t$. Therefore, $t$ is not visible. This  also implies that $x_1$ and $x_2$ are children of $t$.

	Assume $p(t)$ is the start node of $P_{t3}$ and $P_{t3} = P_{p(t)1}$.
	We further prove that $P_{t3}$ consists of only an edge  $(p(t), t)$ in $N$.
	
	Assume on the contrary there are nodes between $p(t)$ and $t$ in $P_{p(t)1}$.  We consider the parent $y$ ($\neq p(t)$) of $t$ in the trivial path $P_{p(t)1}$.  If $y$ is a reticulation node, that $t$ is not visible implies that $y$ is also not visible, a contradiction. Hence $y$ must be a tree node in $N$. We consider the following two cases.
	
	{\bf Case 1}. $y$ is equal to $w_j$ or equal to the other parent of the internal node of $C_j$ for some $j\in\{1, 2\}$. 
	
	Without loss of generality, we may assume $j=1$ (Figure~\ref{Fig7}a and b). This implies that there is no non-cross edge jumping over the cross path 
	$C_1$ and there is a reticulation node $z_1$ in $C_1$. 
	
	When $y=w_1$, let $(z'_1, z_1)$ be the edge removed from $z_1$ in the first stage. Since $z'_1$ is a parent of $z_1$, if $z'_1$ is  below $y$, it must be below $x_2$.  When $y\neq w_1$, $w_1$ is  below $x_2$ if it is below $y$.

	Similarly, $w_2$ is below $x_1$ and thus below $z_1$ if it is below $y$.  
	
	The set of reticulation nodes $\{z_1, x_2\}$ and  $y$ and satisfy the condition in Lemma~\ref{lemma22}, so there is path from $\rho$ to $\ell$ that avoids $y$ for any leaf $\ell$ below $z_1$ and $x_2$. If $\ell$ is below neither $z_1$ nor $x_2$, it is not below $y$, and $P_{T}(\rho, \ell)$ does not pass through $y$. Hence, $y$ is not visible.
	
	{\bf Case 2}. $y$ is neither $w_j$ nor the other parent of the internal node of $C_j$ for each $j=1,2$ (Figure~\ref{Fig7}c). 
	
	In this case, for each $j=1,2$,  $x_j$ is either below $x_{j'}$ or  there is a path from $\rho$ to $w_j$ that avoids $y$. Applying Lemma~\ref{lemma22} on the set of reticulations $\{x_1, x_2\}$ and $y$, we conclude that there is a path from $\rho$ to $\ell$ that avoids $y$ for any leaf $\ell$ below $x_1$ or $x_2$. Clearly,  for any leaf $\ell$ not below $x_1$ or $x_2$, $P_T(\rho, \ell)$ avoids $y$. Therefore,  $y$ must be not visible. That $y$ and $t$ are two consecutive nodes and not visible contradicts that $N$ is nearly-stable.
	
	After proving that the path $P_{p(t)1}$ is actually an edge 
	$(p(t), t)$, we now prove that $w(P_{p(t)2})\leq 2$.  Assume on the contrary $w(P_{p(t)2})\geq 3$. Then, the child $f_3$ of $p(t)$ in $P_{p(t)2}$ must be a reticulation node (Figure~\ref{Fig7}d). Then, the set of reticulations 
	$\{f_3, x_1, x_2\}$ and $p(t)$ satisfy the conditions in Lemma~\ref{lemma22}, so there exists a path 
	from $\rho$ to $\ell$ that does not pass through $p(t)$ for any leaf $\ell$ below $p(t)$ in $T$. For any leaf $\ell$ not below $p(t)$, the tree path $P_{T}(\rho, \ell)$ does not pass through $p(t)$. Hence, $p(t)$ is not visible. That  
	$p(t)$ and $t$ are not visible contradicts that $N$ is nearly-stable network. 
	
		(d) If $P_{t3} = P_0$, then $t$ is an ancestor of any degree-3 node in $T$. Since $N$ is acyclic, there does not exist any cross path $C \in \mathcal{P}$ from $w$ to $x$, such that $x \in P_0$ while $w \in P_i$ for $i>0$. Hence $w(P_0)=0$.
	
	If the weight of $P_{t1}$ and $P_{t2}$ are both 2, and if $w_i$ is the start node of the cross path $C_i$ that enters $P_{ti}$ for $i=1,2$, either $w_1$ or $w_2$ is a node in $P_0$. Following the proof of fact (c), we conclude that  $t$ and its parent in $P_0$ are both not visible, a contradiction.
	\qed
\end{proof}

\begin{figure}[!t] 
	\begin{center}
		\includegraphics*[width=0.9\textwidth]{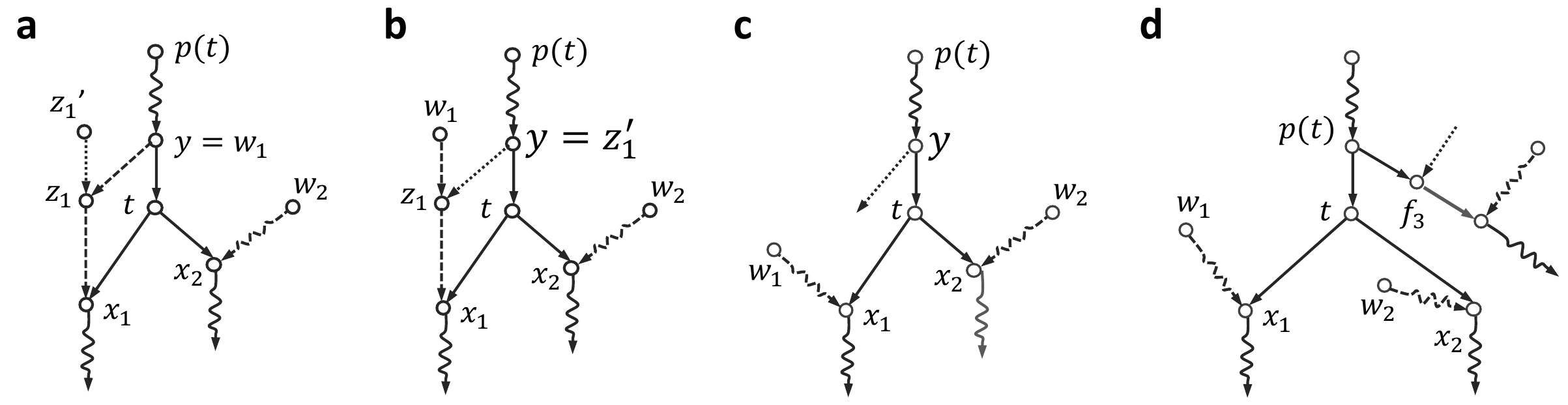}	
	\end{center}
	\caption{The four cases that are considered in the proof of the fact (c) in Proposition~\ref{NS_4}. 
	The path from  $w_i$ to $x_i$  is a cross path ending at $x_i$ in a trivial path $P_{ti}$ leaving $t$ for $i=1, 2$. ({\bf a})  $w_1$ is an internal node in $P_{t3}$.
		({\bf b})  $y$ is the other parent of the unique internal node of the cross path from  $w_1$ to $x_1$. ({\bf c})  neither $w_1$ nor  $w_2$ is  equal to $y$  and has a common child with $y$. (\textbf{d})  $p(t)$ is the parent of $t$ and has a reticulation node $f_3$ as the other child. Solid arrows and curves represent the edges and paths in $T$,  round dot arrows represent edges in $E$ that were removed to form $N'$, and square dot arrows and curves represent the edges and paths that were removed to transform $N'$ to $T$. 
		\label{Fig7}}
\end{figure}

	\begin{theorem} 
		Let $N$ be a nearly-stable network  with $n$ leaves. Then, 
		$|\mathcal{R}(N)| \leq 3(n-1). $
	\end{theorem}
	\begin{proof}
		Let $V$ denote the set of internal nodes of degree 3 in $T$, and let
		$$ V_i = \{v \in V \;|\; w(P_{v1}) + w(P_{v2}) = i \}.$$
		
		For any $v\in V$, we define  $p(v)$ to be the start node of the trivial path $P_{v3}$ that enters $v$. By Proposition~\ref{NS_4} (c) and (d), that $v \in V_4$ implies $p(v) \in V_0 \cup V_1 \cup V_2$. Additionally,  there are at most two different nodes $v'$ and $v''$ in $V_4$ such that $p(v')=p(v'')\in  V_0$, as there are only  two trivial  paths leaving a degree-3 tree node in $T$;
		for different $v'$ and $v''$, if $p(v')$ and $p(v'')$ are in $ V_1 \cup V_2$, then $p(v')\neq p(v'')$.  Taken together, the two facts imply that   $|V_4| \leq 2 |V_0| + |V_1| + |V_2|$. 
		
		Since $w(P_0)=0$, 
		$$\begin{array}{rll}
		|\mathcal{R}(N)| &= &\sum _{v\in V}[w(P_{v1})+w(P_{v2})]\\
		&=&     |V_1|+2|V_2|+3|V_3|+4|V_4| \\
		&\leq&  2|V_0|+2|V_1|+3|V_2|+3|V_3|+3|V_4|\\
		&\leq&   3(|V_0|+|V_1|+|V_2|+|V_3|+|V_4|)\\
		&=&      3(n-1).
		\end{array}$$		\qed
	\end{proof}

	\section{ Stable-child network}
	\label{sec:StableChild}
	
	 The stable-child network shown in Figure~\ref{Fig8}a has as many as $7(n-1)$ reticulation nodes. In this section, we shall prove that a stable-child network can have  $7(n-1)$ reticulation nodes at most.

We first transform a stable-child network to a reticulation visible network and then to a binary tree with the same leaves by removing some edges into reticulations nodes.

	 \begin{figure}[!t] 
	 	\begin{center}
	 		\includegraphics*[width=0.6\textwidth]{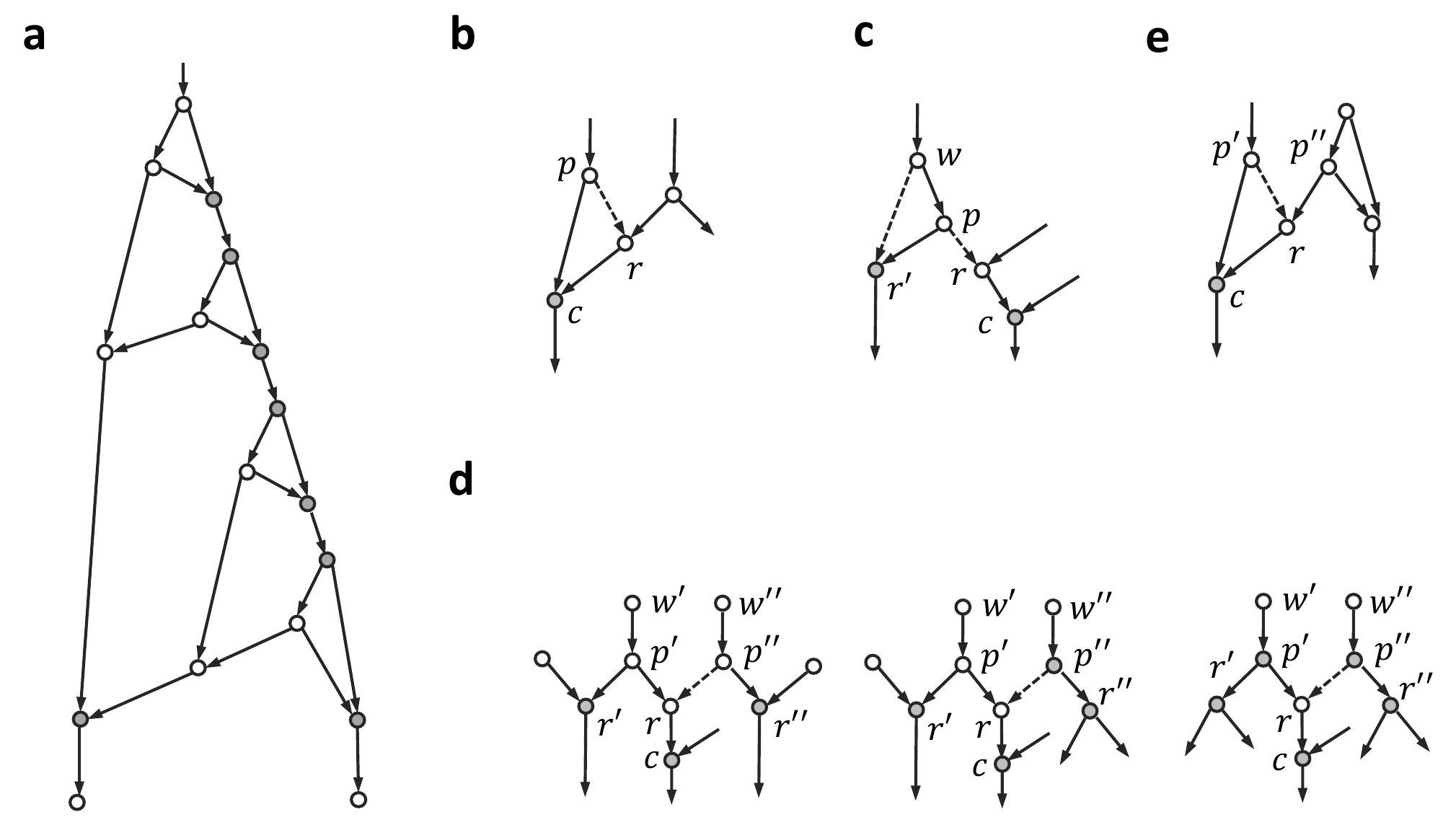}	
	 	\end{center}
	 	\caption{ ({\bf a}) A stable-child network with 2 leaves that has as many reticulation nodes as possible. 
In a stable-child network, there are three possible  local structures at a reticulation node $r$ if it is not visible: 
 ({\bf b}) $r$ and its child $c$ have a common parent. ({\bf c}) $r$ has a parent $p$ and  a sibling $r'$ under $p$ such that $r'$ and $p$ has a common parent $w$. ({\bf d}) Neither ({\bf b}) nor ({\bf c}) is true. 
 ({\bf e}) Both ({\bf b}) and ({\bf C}) occur at the same time \label{Fig8}}
	 \end{figure}

	 \begin{proposition}\label{SC_1}
	 	Let $N$ be a stable-child network. There is a set of edges $E$ such that {\rm (1)} $N-E$ is a subdivision of a reticulation visible network over the same leaves as $N$, and {\rm (2)} $E$ contains exactly an incoming edge for a reticulation node if it is not visible in $N$.
	 \end{proposition}
	 \begin{proof} 
	 	For  a reticulation node $r$ that is not visible in $N$, its unique child must be a visible reticulation node.  Furthermore, since each node has a visible child, 
	 	its parents both have a visible child other than $r$ and are both a tree node. To transform $N$ into a reticulation visible  network, we will delete one or two edges around a reticulation if it is not visible.

For each reticulation node $r$ that is not visible, we consider the following three cases.
If  $r$ and its unique child $c$ have a common parent $p$ (Figure~\ref{Fig8}b), then $(p, r)$ is removed.
	 	
If $r$ and its child $c$ do not have a common parent, but $r$ has a reticulation sibling $r'$ such that    the parent $w$  of  the common parent $p$ of $r$ and $r'$ is the other  parent of $r'$ (Figure~\ref{Fig8}c),    $(p,r)$ and $ (w, r')$ are then deleted at the same time.  
	 	
When neither occurs (Figure~\ref{Fig8}d), 
we arbitrarily select an incoming edge of $r$ to remove. 
	 	
	 	The edges removed in the above process is called \emph{SC-edges}.  Each  SC-edge is from  a tree node to  a reticulation node. A SC-edge is \textit{concealed} if the head is a visible reticulation node; it  is \textit{revealed} otherwise. Note that a concealed SC-edge is deleted only when the case shown in  Figure~\ref{Fig8}c is satisfied. Therefore, a concealed SC-edge jumps over the associated  revealed SC-edge that is removed at the same time. It is not hard to see that the SC-edges that are removed when different reticulation nodes are considered are  different. Let $E$ be the set of SC-edges.

	 	First,  we only deleted an incoming edge for each reticulation node and did not delete the incoming edge for each tree node, so the resulting network $N-E$ is connected. Second, $N-E$ has the same leaves  as $N$. The reasons for this include that (i) we do not remove any outgoing edge of a reticulation node, and (ii)  for any tree node $t$, if an outgoing edge of it  is removed,  the other outgoing edge   enters another tree node and thus  has never not been removed.
	 	
	 	Now, we show that $N-E$ is a subdivision of a binary reticulation visible network. Since we had deleted  an incoming edge for a reticulation node if it is not visible, all the remaining reticulation nodes are visible. $N-E$ is reticulation visible. We  can also show that  there are no two internally disjoint paths from a common tree node to a common reticulation node in which 
each internal node is of degree 2, implying that $N-E$ is a subdivision of a binary reticulation visible network.
	 	
	 	Assume on the contrary there are two internally disjoint path $P_1$ and $P_2$ between $u$ and $v$ such that their internal nodes are all of  degree 2. If neither $P_1$ nor $P_2$ is a single edge, then the two children of $u$ in $P_1$ and $P_2$ are both not visible, 
contradicting that $N$ is a stable-child network. Therefore, either $P_1$ or $P_2$ is a single edge from $u$ and $v$.  

Without loss of generality, we may assume that $P_2$ is equal to the edge $(u,v)$. 
 According to  the three rules which we used to remove the edges in $E$,  if an incoming edge of a node  is removed, its child in $N-E$ is visible in $N$. 
This implies that $N-E$ does not contain a path consisting of two or more degree-2 nodes that are not visible in $N$. Therefore, $P_1$ has exactly one internal node $x$.  If $x$ is a tree node in $N$, then, we removed an outgoing edge of  $x$ according to either  the second or third case. In the former case, we remove $(u,v )$ at the same time. In the later case, $(u, v)$ does not exist in $N$. This is impossible. 
When $x$ is a reticulation node, we removed an incoming edge of it. Again, 
 the edge $(u, v)$ does not exist in $N$ in each possible case, a contradiction.

	 We have proved that  $N-E$ is a subdivision of a binary network.
	 	\qed
	 \end{proof}

	 Let $N' = N-E$ be the  subnetwork obtained after the removal of the edges in $E$. 
	 $N'$ is a subdivision of a reticulation visible network.  By Theorem~\ref{stable_1},  there exist  a set of  paths $\mathcal{P}$ such that (i)  $T:=N - E - \cup_{P\in\mathcal{P}} \mathcal{IV}(P)-\cup_{P\in\mathcal{P}} \mathcal{E}(P)$ is a subtree of $N$ with the same leaves and (ii) all the internal nodes in each path in $\mathcal{P}$ are of degree 2.  

Again, we use  $P_0, P_1,...,P_{2n-2}$ to  denote the  trivial paths in $T$, where $P_0$ denotes the trivial path starting with  $\rho(N)$. As in the last section, a path in $\mathcal{P}$ is called a non-cross path if  its start and end nodes are both in $P_j$ for some $j$; it is called a cross path  otherwise.

	\begin{lemma} \label{SC_2}
Let $P$ be a path in $\mathcal{P}$.

{\rm (a)} Every internal node  in $P$ is not visible in $N$.	

{\rm (b)} If $P$ is a non-cross path, it is simply an edge.	
		
{\rm (c)} If  $P$ is a cross path and ends at a node $x$ in the trivial path $P_j$, every node 
 between the start node of $P_j$ and $x$ in $T$ is not visible in $N$.
		
		
{\rm (d)} If $P$ is a cross path and there is a non-cross path $P'$ jumping over it, then either $P$ is an edge or   the start node of $P$ is the parent of the end node of $P'$ in $T$.
	
	\end{lemma}	
	\begin{proof}
		(a) and (b) are essentially the restatement of the fact (a) and (c) in Lemma~\ref{NS_2}.

		
		(c) Let $y$ be a node between the start node of $P_j$ and $x$ in $T$. For any leaf $\ell$ that is not below $x$ in $T$, $P_T(\rho,\ell)$ is a path that does not pass $x$ and hence $y$. For any leaf $\ell$  below $x$ in $T$, the path $P_T(\rho, w) + P + P_T(x, \ell)$ avoids $y$, as $w$ is the start node of $P$ in a trivial path different from $P_j$. Hence,  $y$ is not visible in $N$.

		(d) By (b), $P'$ is simply an edge $(u, v)$ in $N'$. Let $P$ start at a node $w$.  If neither $P_{T}(w, v)$ nor  $P$ is a single edge,  the two children of $w$ in $P$ and $P_T(w, v)$ are both not visible, contradicting that $N$ is stable-child.
		\qed
	\end{proof}

Let $r$ be a reticulation node and not visible in $N$. Then,  
a revealed SC-edge $e_r$ was removed from $r$ to obtain $N'$ from $N$.  We define the cost $c(r)$ of $r$ to be:
\begin{equation*}
	c(r)=\begin{cases}2 & \mbox{if a concealed SC-edge is associated  with $e_r$}, \\1 &\mbox{otherwise.}
	\end{cases}
	\end{equation*}

Recall that  $U_s = parent(s) \cap \mathcal{UR}(N)$.
	We can define the cost of a cross path $C\in \mathcal{P}$ from $w$ to $x$ as follows:
	\begin{eqnarray}
	c(C)=\begin{cases}2 + \displaystyle{\sum_{r \in U_x \cup U_v} c(r)} & \mbox{if a non-cross edge $(u,v)$ jumps over $w$}, \\1+ \displaystyle{\sum_{r \in U_x} c(r)}  & \mbox{otherwise.}
	\end{cases}
	\label{SCformula}
	\end{eqnarray}
	We further charge the cost of $C$ to the trivial path $P_i$ to which $x$ belongs and call it the weight of $P_i$, written $w(P_i)$.
	If there is no cross path entering $P_i$, the weight of $P_i$  is set to be 0.
	
As for nearly-stable networks,  we have that
	$|\mathcal{R}(N)| \leq \sum_{i=0}^{2n-2}w(P_i).$

	For an internal node $t$ with degree-3 in $T$,  we still use  $P_{t1}$ and $ P_{t2}$ to denote  the trivial paths leaving $t$ and $P_{t3}$ to denote the trivial path entering $t$.

	\begin{proposition}\label{SC_3}
		For each internal node $t$ of degree 3 in $T$,
		
		{\rm (a)} $w(P_{tj}) \leq 6$, $j=1, 2$.
		
		{\rm (b)} $w(P_{t1}) + w(P_{t2}) \leq 10$.
		
		{\rm (c)} If $P_{t3} \neq P_0$ and $w(P_{t1}) + w(P_{t2}) \geq 8$, then $w(P_{t3}) = 0$. Moreover, assume $p(t)$ is the start node of $P_{t3}$ and $P_{t3} = P_{p(t)1}$. Then $w(P_{p(t)2}) \leq 4$.
		
		{\rm (d)} If $P_{t3} = P_0$, then $w(P_{t3}) = 0$ and $w(P_{t1}) + w(P_{t2}) \leq 7$.
		
	\end{proposition}
	\begin{proof}
		(a) We will only prove that $P_{t1}\leq 6$. Let $C_1$ denote the cross path sending at a node $x_1$ in $P_{t1}$. Let $w_1$ be the start node of $C_1$ in a trivial path different from $P_{t1}$. Note that  $U_s = \mathcal{UR}(N) \cap parent(s)$. 
		
		If there is no non-cross edge jumping over $w_1$, then there are at most 2 elements in $U_{x_1}$, and each element can have two unit cost at most. Thus,   by Eqn.~(\ref{SCformula}), $w(P_{t1})\leq 5$. 
		
		If there is a non-cross edge $(u_1,v_1)$	jumping over $w_1$. By the fact (b) in Lemma~\ref{SC_2}, $U_{v_1}$ is empty or a singleton. Moreover, by the fact (d) in Lemma~\ref{SC_2} (d), either $C_1$ is an edge, or $(w_1,v_1)$ is an edge in $T$. If $C_1$ is an edge, then $|U_{x_1}| \leq 1$.  If  $(w_1,v_1)$ is an edge, $|U_{v_1}| = 0$. Both implies that $|U_{x_1}| + |U_{v_1}| \leq 2$. Therefore,  by Eqn.~(\ref{SCformula}),  $w(P_{t1}) \leq 2 + 2(|U_{x}| + |U_{v}|) \leq 6$.
		
		We remark that, if the parent of $x_1$ in $T$ is not a reticulation node in $N$, then $|U_{x_1}| + |U_{v_1}| \leq 1$, and therefore $w(P_{t1}) \leq 4$. Equality holds only if there is a non-cross edge jumping over $x$.
	
		(b) If  $w(P_{t1})\neq 0$ and $w(P_{t2})\neq 0$, we assume that  the cross path $C_i$  ending at  a node $x_1$ in $P_{ti}$ starts at $w_i$ for $i=1, 2$. By the fact (c) of Lemma~\ref{SC_2}, every internal node in $P_T(t,x_1)$ and $P_T(t,x_2)$ is not visible. If there is a node between $t$ and $x_i$ for each $i$, then the two children of $t$ are not visible in $N$, contradicting that $N$ is stable-child. So  $t$ is the parent of either $x_1$ or $x_2$. Without loss of generality, we may assume that  $t$ is the parent of $x_1$ in $T$.  By the remark in the end of  the proof of (a), $w(P_{t1}) \leq 4$ and hence $w(P_{t1}) + w(P_{t2}) \leq 4 + 6 = 10$.

		(c) Assume that $P_{t3} \neq P_0$ and  $p(t)$ be the start node of $P_{t3}$ such that $P_{t3} = P_{p(t)1}$. If  $w(P_{t1}) + w(P_{t2}) \geq 8$, by  (a), the weights of $P_{t1}$ and $P_{t2}$ are both not zero. Hence,  there is a cross path $C_i$ ending at a node $x_i$  in $P_{ti}$ and starting at a node $w_i$ in a trivial path different from 
$P_{ti}$ for each $i = 1,2$.

		We first show that $w_i$ either  (i) below $x_{3-i}$, or (ii) neither in $P_{t3}$ nor below $t$ for $i=1$ and $2$.
Without loss of generality, we assume $i=1$. 
		
		{\bf Case 1}. $w_1$ is in $P_T(t,x_2)$ (Figure~\ref{Fig9}a). 
		
		If there is a non-cross edge $(u_1,v_1)$ jumping over $w_1$,  $v_1$ is either in $P_T(w_1,x_2)$ or  $v_1$ is below $x_2$ in $P_{t2}$. The former implies that $v_1$ is not visible, whereas the latter implies that $x_2$ is not visible.
This contradicts that both $v_1$ and $x_2$ are visible in $N'$. 
		
		Since $w_1$ is an internal node between $t$ and $x_2$ in $T$,  then $t$ is the parent of $x_1$ in $T$. 
By the fact (a) and (c) of Lemma~\ref{SC_2}, each  internal node in $C_1$ or $P_T(w_1,x_2)$ is not visible. Thus,  $(w_1,x_1)$ or $(w_1,x_2)$ is an edge. Otherwise,  the two children of $w_1$ are not visible, contradicting $N$ is stable-child network. 
		
		If $(w_1,x_1)$ is an edge in $N$, then $w(P_{t1}) = 1$ and hence $w(P_{t1}) + w(P_{t2}) \leq 7$, a contradiction. If $(w_1,x_2)$ is an edge, then $w(P_{t1}) \leq 3$. Since $w_1$ is not a reticulation, by the remark at the end of the proof of (a), 
 $w(P_{t2}) \leq 4$. Therefore, $w(P_{t1}) + w(P_{t2}) \leq 7$. This is impossible.

{\bf Case 2}. $w_1$ is a node in the path $P_{t3}$. 
		
		Without loss of generality, we can assume $w_1$ is lower than $w_2$ if $w_2$ is also in $P_{t3}$. We claim that $t$ and all the  internal nodes in $P_T(w_1,t)$ are not visible. Let $u$ be either $t$ or an internal node in $P_T(w_1,t)$.  If $w_2$ is in $P_T(t, x_1)$, the case is symmetric to Case 1. So there are two cases to consider: either 
		$w_2$ is below $x_1$ (Figure~\ref{Fig9}b), or
		$w_2$ is not below $w_1$ (Figure~\ref{Fig9}c). 
		In both cases, the reticulation node set $\{x_1,x_2\}$ are below the node $u$, and satisfies the condition in Lemma~\ref{lemma22}, so $u$ is not visible with respect to each leaf $\ell$ below either $x_1$ or $x_2$. For any leaf $\ell$ not below $x_1$ or $x_2$,   the path $P_T(\rho, \ell)$ avoids $u$. Hence $u$ is not visible in $N$.
		
		There are some observations from this result. First, there is no non-cross edge $(u_1,v_1)$ jumping over $w_1$, otherwise $v_1$ is not visible. 
		Second, a child of $w_1$ in $P_{t3}$ is not visible, and so the cross path $C_1$ is simply an edge. 
Otherwise by  the fact (a) of Lemma~\ref{SC_2}, the two children of $w_1$ are both not visible in $N$. 
		
		By (b) of  Lemma~\ref{SC_2}, either $(t, x_1)$ or $(t, x_2)$ is an edge in $T$. If $t$ is the parent of $x_1$ in $T$, then $w(P_{t1})=1$ according to Eqn.~(\ref{SCformula}).  By (a), $w(P_{t1}) + w(P_{t2}) \leq 7$. 
 If $t$ is the parent of $x_2$,   $w(P_{t1}) \leq 3$ and $w(P_{t2}) \leq 4$ according to  the remark at the end of the proof of  (a). Taken together, both facts imply that $w(P_{t1}) + w(P_{t2}) \leq 7$, which contradicts the assumption that $w(P_{t1}) + w(P_{t2}) \geq 8$. 
		
		{\bf Case 3}. $w_1$ is below $x_2$ and $w_2$ is below $x_1$ (Figure~\ref{Fig9}d).
		
This case is impossible since there is a cycle in $N$, contradicting $N$ is acyclic.

		To sum up, $w(P_{t1}) + w(P_{t2}) \geq 8$ implies that either (i) or (ii) is true. But in both cases, if we let $u$ be either $t$ or any internal node in $P_{t3}$, the set of reticulations $\{x_1, x_2\}$ are below $u$ and satisfies the condition in Lemma~\ref{lemma22}. Therefore, $t$ and any internal node of $P_{t3}$ are not visible.
		
		There are two observations from this result. First, $w(P_{t3}) = 0$ because there is no cross path that ends at $P_{t3}$.  (Otherwise the cross path enters $P_{t3}$ at a  reticulation that is not visible in $N$.)  Second, the child of $p(t)$  in  $P_{t3} = P_{p(t)1}$ is not visible. 
		
		Clearly,  $w(P_{p(t)2}) = 0$ if there is no cross path that ends at $P_{p(t)2}$.  Assume there is a cross path $C_3$ from $w_3$ to $x_3$ with $x_3$ in $P_{p(t)2}$. By (c) of Lemma~\ref{SC_2}, each  internal node in $P_T(p(t),x_3)$ is not visible.  But  the child of $p(t)$ in $P_{p(t)1}$ is not visible. Hence, $(p(t),x_3)$ is an edge in $T$. Then,  by the remark in the end of proof of (a), $w(P_{p(t)2}) \leq 4$.
		
		(d) If $P_{t3} = P_0$, then $t$ is an ancestor of any degree-3 node in $T$. Since $N$ is acyclic, there does not exist any cross path $P(u,v) \in \mathcal{P}$ such that $u \in P_i$ for $i>0$ and $v \in P_0$. Hence,  $w(P_0)=0$.  

 If $w(P_{t1}) + w(P_{t2}) \geq 8$, then, every node in $P_0$ is not visible in $N$,  shown in 
(c). This contradicts that $P_0$ contains the network root $\rho(N)$ and $\rho(N)$ is visible with respect to each leaf in $N$. 
		\qed				
	\end{proof}
	
				\begin{figure}[!t] 
					\begin{center}
						\includegraphics*[width=0.8\textwidth]{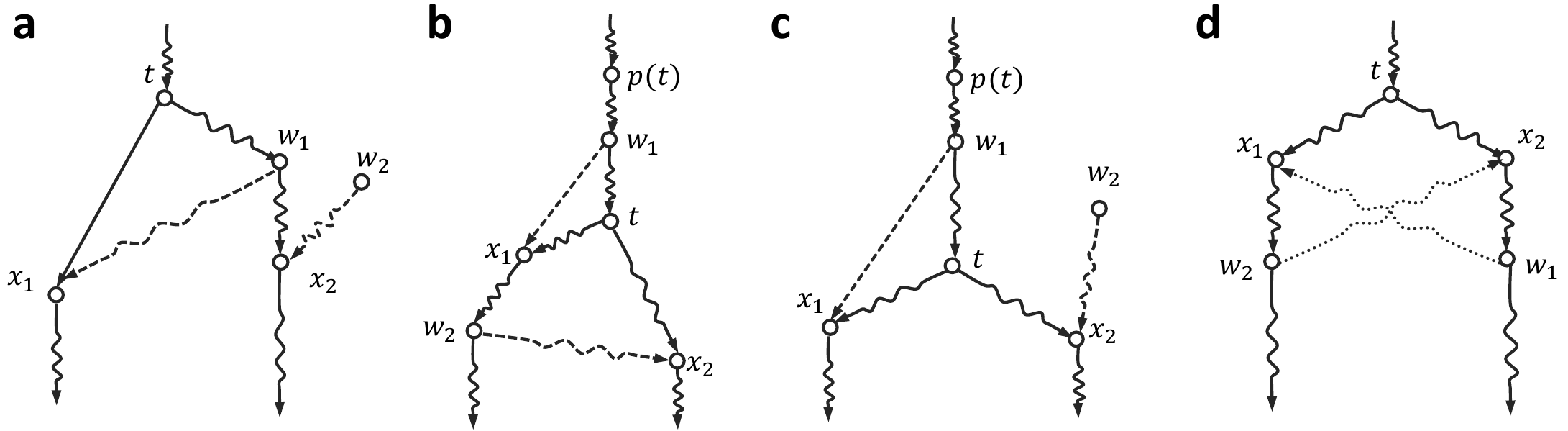}	
					\end{center}
					\caption{Four cases that are considered in the proof of Proposition~\ref{SC_3}.  
(\textbf{a}) $w_1$ is between $t$ and $x_2$. (\textbf{b}) $w_1$ is the trivial path $P_{t3}$ entering $t$, and $w_2$ is below $x_1$. (\textbf{c}) $w_1$ is in $P_{t3}$, and $w_2$ is not below $x_1$. (\textbf{d}) $w_1$ is below $w_2$ and $w_2$ is also below $x_1$. This is impossible in a network   \label{Fig9}}
				\end{figure}

	\begin{theorem}
		Let $N$ be a stable-child network with $n$ leaves. Then, 
		$|\mathcal{R}(N)| \leq 7 (n-1).$
	\end{theorem}
	\begin{proof}
				Let $V$ denote the set of the $(n-1)$ internal nodes of degree 3 in $T$. Define
				$$ V_i = \{v \in V | w(P_{v1}) + w(P_{v2}) = i \}. $$
By Proposition~\ref{SC_3} (b), $w(P_{v1}) + w(P_{v2})\leq  10$. Hence,  
$V = \uplus_{i=0}^{10} V_i,$ and thus 
$\sum^{10}_{i=0}|V_i|=n-1$.

Let $p(v)$ be the start node of the trivial path entering $v$ in $T$. By Proposition~\ref{SC_3} (c) and (d),  if  $v \in V_j$, then $p(v) \in V_0\cup V_1\cup V_2\cup V_3\cup V_4$  for each $j \in \{8,9,10\}$.  By Proposition~\ref{SC_3} (c), under the mapping $p(\cdot)$, at most two nodes in $V_8\cup V_9\cup V_{10}$ are mapped to the same node in $V_0$, and only one node can be mapped to $V_1\cup V_2\cup V_3\cup V_4$. Thus,  
				\begin{equation*}
				|V_8| + 2|V_9| + 3|V_{10}| 
 \leq 3(|V_8| + |V_9| + |V_{10}|)\leq 6 |V_0| + 3(\sum_{i=1}^4 |V_i|).
				\end{equation*}
Since 	$w(P_0)=0$, 	the above inequality implies that 
			\begin{eqnarray*}
				|\mathcal{R}(N)| &= &\sum _{v\in V}[w(P_{v1})+w(P_{v2})]\\
				&=&     \sum_{i=0}^{10} i |V_i| \\
				&\leq&  |V_1|+2|V_2|+3|V_3|+4|V_4| +  7 \sum_{i=5}^{10}  |V_i| +  |V_8| + 2|V_9| + 3|V_{10}|) \\	
	&\leq &  |V_1|+2|V_2|+3|V_3|+4|V_4| +  7 \sum_{i=5}^{10}  |V_i| + 6 |V_0| +  3\sum_{i=1}^4 |V_i|  \\		
				&=& 6|V_0|+4|V_1|+5|V_2|+6|V_3|+7|V_4| +  7 \sum_{i=5}^{10}  |V_i| \\
				&\leq&   7\sum_{i=0}^{10} |V_i|\\
				&=&      7(n-1).
		\end{eqnarray*}
	\qed	
	\end{proof}

\section{ Conclusion}
\label{conc}

We have established the tight upper bounds for the sizes of galled,  nearly-stable, and stable-child networks. Since the number of internal tree nodes is equal to the number of leaves plus the number of reticulation nodes in a binary network,  we can summarize our results in Table~\ref{table1}.  Without question, these tight bounds provide insight to the study of combinatorial and algorithmic aspects of the  network classes defined by visibility property. 

\begin{table}[h]
\caption{The tight upper bounds on the sizes of binary networks with $n$ leaves defined by  visibility property. The bound for reticulation visible network is found in Bordewich and Semple(2015). 
\label{table1}}
\begin{tabular}{lcc}
\hline
   & No. of reticulation nodes & No. of internal tree nodes \\
\hline
Galled Network &  $2(n-1)$ & $3(n-1)$\\
Reticulation visible network & $3(n-1)$ & $4(n-1)$\\
Nearly-stable network & $3(n-1)$ & $4(n-1)$\\
Stable-child network & $7(n-1)$ & $8 (n-1)$\\
\hline
\end{tabular} 
\end{table}

\end{document}